\documentclass[11pt,a4paper]{article}
\usepackage[utf8]{inputenc}
\usepackage{amsmath, amssymb, amsthm}
\usepackage{mathtools}
\usepackage{geometry}
\usepackage{hyperref}
\usepackage{booktabs}
\usepackage{graphicx}
\usepackage{algorithm}
\usepackage{algpseudocode}
\usepackage[numbers]{natbib}
\usepackage{xcolor}
\usepackage{tikz}
\usepackage{enumitem}

\hypersetup{
    colorlinks=true,
    linkcolor=blue!60!black,
    citecolor=blue!60!black,
    urlcolor=blue!60!black
}

\newtheorem{theorem}{Theorem}[section]
\newtheorem{proposition}[theorem]{Proposition}

\newtheorem{corollary}[theorem]{Corollary}
\newtheorem{definition}[theorem]{Definition}
\newtheorem{remark}[theorem]{Remark}
\theoremstyle{definition}
\newtheorem{auditimp}{Audit Implication}[section]
\newtheorem{designprinciple}{Design Principle}[section]
\newtheorem{heuristic}{Heuristic}[section]
\theoremstyle{definition}
\newtheorem{assumption}[theorem]{Assumption}
\theoremstyle{plain}
\theoremstyle{plain}

\DeclareMathOperator{\tr}{tr}
\DeclareMathOperator{\proj}{proj}

\title{\textbf{Telemetry and Concealment in Self-Adapting Generative AI:\\Logging Architecture, Adversarial Model Hiding,\\and the Limits of Detection}}

\author{
  Sriram Nagaraj \footnote{dr.sri.nagaraj@gmail.com}
}

\date{}

\begin{document}

\maketitle

\begin{abstract}
Model risk management (MRM) guidance assumes a static model lifecycle, in which models are developed, independently validated, and implemented without further autonomous modification. Continually self-adapting generative AI systems --- models that update their own weights during production deployment --- fundamentally violate this assumption and render point-in-time validation inadequate. This paper addresses the resulting governance problem in two parts.

\emph{Part I} develops a rigorous telemetry architecture for such models, operating simultaneously in discrete and continuous time. We establish the \emph{Minimal Sufficient Statistic} for audit purposes, construct a tamper-evident Merkle chain for discrete weight sequences, derive the appropriate continuous-time generalization via the It\^{o} formula, and propose event-driven logging via KL divergence stopping times that is both computationally tractable and meaningful for validation. We further embed these results within a Lyapunov stability framework, deriving both mean-stability and almost-sure path-wise stability conditions for the weight process, and show how the CUSUM test \citep{basseville1993detection} provides a statistically grounded mechanism for tier-based governance escalation.

\emph{Part II} asks what happens when the model owner is adversarial. A firm deploying such a model has strong economic incentives to conceal learning updates that would trigger mandatory validation review. We formalize this as the \emph{Model Hiding Problem} and provide a systematic taxonomy of six distinct attack strategies against the Part~I architecture, spanning discrete and continuous time, with a formal countermeasure for each grounded in cryptography, statistical hypothesis testing, or stochastic process theory. Two of the natural continuous-time defenses are shown, as \emph{negative} results, not to work: the first variation of a diffusion is infinite, so a power-variation ratio measures sampling frequency rather than drift, and a gradient probe is circular, re-detecting only what the behavioral oracle already sees. A third continuous-time attack hides adversarial drift inside stochastic learning noise; here a path-wise check against the declared quadratic variation (an envelope from the Law of the Iterated Logarithm) does close the gap.

Far from weakening the framework, the two failed defenses sharpen its central result: we prove that at any finite detection resolution $\delta>0$ the set of undetectable weight perturbations contains a ball of strictly positive radius $\epsilon^*(\delta)$, so no telemetry-only regime can detect arbitrarily small adversarial drift. We do not derive a closed form or a rate for $\epsilon^*(\delta)$, and in particular no dependence on telemetry bandwidth: what is established is strict positivity and monotonicity in $\delta$, which is what the dual-regime conclusion requires. A separate counting argument suggests why a residual should exist in high dimension, but it is stated as a heuristic rather than a theorem, and it is not tested at the scale simulated here. Together the two parts establish a dual-regime architecture in which continuous telemetry is necessary but not sufficient, narrowing---but never replacing---periodic invasive audit. The framework is model-architecture-agnostic and is designed to satisfy the \emph{Effective Challenge}, \emph{Conceptual Soundness}, and \emph{Outcome Analysis} pillars of the model risk management guidance.
\end{abstract}

\tableofcontents
\newpage

\section{Introduction}

The principles of model risk management (MRM) enunciated in \citep{sr117} establish a lifecycle framework for models used in consequential decisions: development, independent validation, and production implementation \citep{sr117}. Implicit in this framework is the assumption that the model in production is the model that was validated --- that is, the mapping from inputs to outputs does not autonomously change over time.

Generative AI systems with online or continual learning capabilities violate this assumption by design. A large language model (LLM) fine-tuned via reinforcement learning from human feedback (RLHF) on production traffic, or a deep neural network updated continuously via stochastic gradient descent (SGD) on streaming data, is a \emph{different} model at each inference step. Traditional validation, which certifies a model at a fixed point in time, provides no ongoing guarantee about the model's behavior, risk characteristics, or alignment with its approved objective.

This creates what we term the \emph{Continuous Validation Gap}: the period between the most recent formal validation and the current model state. As learning rates, data volumes, and architectural complexity increase, this gap widens, and the distance between the validated model and the production model may become unacceptable from a governance standpoint.

Closing that gap requires two things, and this paper supplies both. The first is an architecture that makes adaptation \emph{observable}: a faithful, tamper-evident, computationally tractable record of how a model moves. The second is an account of what happens when the model owner does not wish to be observed. A firm that must halt a model pending revalidation whenever it reports material adaptation has an incentive to adapt without reporting; any telemetry regime that assumes truthful transmission is therefore incomplete. Part~I of this paper builds the observability architecture on the assumption of faithful reporting; Part~II removes that assumption, treats the model owner as strategic, and asks how much concealment a telemetry regime can detect --- and, more importantly, how much it provably cannot.

\subsection{Scope and Contribution}

This paper addresses three interconnected questions:
\begin{enumerate}[label=(\roman*)]
    \item \emph{What is the minimal information that must be logged to provide an auditable record of a self-adapting model's trajectory?}
    \item \emph{How should this logging architecture behave in the continuous-time limit, when updates are effectively instantaneous?}
    \item \emph{If the model owner is strategic and wishes to conceal adaptation, which concealment strategies can be detected from the telemetry stream, and which provably cannot?}
\end{enumerate}

Our contributions are as follows:
\begin{itemize}
    \item We formalize the \emph{State Digest} as the minimal sufficient statistic for discrete-time weight trajectory auditing, and show that weight-space distance is a poor proxy for the relevant governance quantity: behavioral divergence.
    \item We construct a tamper-evident \emph{Merkle Weight Chain} that binds each model state to its predecessor and its training data, and show that this is the correct discrete-time formalization of the ``Model Genetic Record.''
    \item We derive the continuous-time limit of the logging architecture via the It\^{o} formula, establishing the \emph{Quadratic Variation} as the canonical telemetry scalar and KL divergence crossing times as the natural event-driven logging mechanism.
    \item We embed the framework within a Lyapunov stability theory, deriving both mean and almost-sure stability conditions for the weight process as a stochastic differential equation (SDE).
    \item We propose a CUSUM-based trigger mechanism that operationalizes the heuristic tier table of prior governance frameworks into a statistically grounded detection rule with explicit false-alarm rate controls.
    \item We formalize the \emph{Model Hiding Problem} and give a complete taxonomy of six concealment strategies against the above architecture, with a formal countermeasure for each --- including two \emph{negative} results showing that the natural continuous-time defenses do not work.
    \item We prove a \emph{detection floor}: at any finite behavioral resolution $\delta>0$ the set of undetectable weight perturbations contains a ball of strictly positive radius $\epsilon^*(\delta)$, so no telemetry-only regime can detect arbitrarily small adversarial drift. We establish positivity and monotonicity of $\epsilon^*(\delta)$, not its magnitude or rate, and we set out the \emph{Dual-Regime Principle} --- a design recommendation, not a theorem --- that this motivates.
\end{itemize}

The remainder of the paper is organized in two parts. \textbf{Part~I (Sections~\ref{sec:background}--\ref{sec:numerical_a})} develops the telemetry architecture under faithful reporting: Section~\ref{sec:background} gives background on the self-adapting model problem; Section~\ref{sec:discrete} develops the discrete-time architecture; Section~\ref{sec:continuous} derives the continuous-time generalization; Section~\ref{sec:lyapunov} establishes the stability framework; Section~\ref{sec:cusum} proposes the CUSUM trigger; Section~\ref{sec:architecture} assembles the integrated architecture; and Section~\ref{sec:numerical_a} reports numerical studies. \textbf{Part~II (Sections~\ref{sec:hiding}--\ref{sec:numerical_b})} removes the faithful-reporting assumption: Section~\ref{sec:hiding} formalizes the Model Hiding Problem; Section~\ref{sec:game} sets out the adversarial game; Sections~\ref{sec:discrete_attacks} and~\ref{sec:continuous_attacks} develop the discrete- and continuous-time attack taxonomies and countermeasures; Section~\ref{sec:residual} establishes the information-theoretic residual; Section~\ref{sec:defensive_arch} describes the integrated defensive architecture; and Section~\ref{sec:numerical_b} reports the corresponding numerical studies. Section~\ref{sec:discussion} discusses both parts together and Section~\ref{sec:conclusion} concludes.

\section{Background}
\label{sec:background}

\subsection{MRM and Its Static Assumptions}

MRM defines a \emph{model} as, roughly, any quantitative method, system, or approach that applies established theories, techniques, and assumptions to process input data into quantitative estimates (\citep{sr117}). The guidance identifies three components of model risk: \emph{fundamental model error} (the model's conceptual design is flawed), \emph{implementation error} (the model is correctly conceived but incorrectly executed), and \emph{inappropriate use} (the model is applied outside its validated scope).

The MRM lifecycle comprises development, validation, and ongoing monitoring. Ongoing monitoring under MRM is designed to detect \emph{model drift} --- degradation in predictive performance over time due to changes in the environment \citep{ben2010theory} --- but it assumes the model itself is fixed; only the environment changes. This is precisely the assumption violated by self-adapting systems, where the model and the environment co-evolve.

\subsection{Self-Adapting Generative AI}

We define a \emph{self-adapting generative model} as a parameterized model $f(\cdot\,;\, W)$ whose weight matrix $W$ is updated continuously during production operation, based on observed data. Formally, at each production step $t$, the model receives a data batch $\mathcal{B}_t$ drawn from a distribution $\mathbb{P}_t$ (which may itself be non-stationary), and updates:
\begin{equation}
    W_{t+1} = W_t - \eta_t \nabla_W \mathcal{L}(W_t, \mathcal{B}_t) + \xi_t
\end{equation}
where $\eta_t$ is the learning rate, $\mathcal{L}$ is a loss function, and $\xi_t$ is a stochastic perturbation term (e.g., from dropout or gradient noise). If we interpret the sequence $\{W_t\}_{t \geq 0}$ as a Markov chain on a high-dimensional weight space $\mathcal{W} \subset \mathbb{R}^d$, then $d$ may be of order $10^9$ to $10^{12}$ for contemporary LLMs. In Section \ref{sec:state_interpretation} we provide a more nuanced interpretation of what the ``effective" dimension of the space we are modeling is. 

The validatory challenge is threefold: (i) storing or transmitting $\{W_t\}$ is computationally infeasible; (ii) the validator cannot, in general, distinguish a model that has drifted significantly from one that has not, without direct access to the weights; (iii) the model's behavior can change dramatically in ways that are not reflected in simple performance metrics.

\section{Discrete-Time Telemetry Architecture}
\label{sec:discrete}

\subsection{The Minimal Sufficient Statistic for Audit}

The central design question is: what information about the weight trajectory $\{W_0, W_1, \ldots, W_T\}$ is \emph{necessary and sufficient} for a validator to certify that the model remained within its validated scope? We answer this in terms of a sufficient statistic relative to a canonical reference set.

\begin{definition}[Canonical Reference Set]
\label{def:reference_set}
Let $\mathcal{R} = \{(x_i, p^*_i)\}_{i=1}^m$ be a fixed set of inputs and their output distributions under the initially validated model $W_0$. $\mathcal{R}$ is chosen at the time of initial validation and held jointly by the firm and the validator. The reference set is said to be \emph{behaviorally sufficient} if it spans the input manifold in a statistically meaningful way.
\end{definition}

\begin{definition}[Behavioral Divergence]
The behavioral divergence between consecutive model states is:
\begin{equation}
    \delta_t = D_{\mathrm{KL}}\!\left(p_{W_{t-1}}(\cdot \mid \mathcal{R}) \;\big\|\; p_{W_t}(\cdot \mid \mathcal{R})\right) = \sum_{i=1}^m \sum_y p_{W_{t-1}}(y \mid x_i) \log \frac{p_{W_{t-1}}(y \mid x_i)}{p_{W_t}(y \mid x_i)}
\end{equation}
\end{definition}

\begin{definition}[Cumulative Drift]
The cumulative behavioral drift from the validated state $W_0$ to the current state $W_t$ is:
\begin{equation}
    \Delta_t = \sum_{s=1}^t \delta_s
\end{equation}
\end{definition}

\begin{remark}
$\Delta_t$ is an additive book-keeping quantity, not a distance: KL divergence is not a metric and obeys no triangle inequality, so in general $\Delta_t$ neither upper- nor lower-bounds the end-to-end divergence $D_{\mathrm{KL}}(p_{W_0}\|p_{W_t})$. For small consecutive steps, however, each increment admits the second-order expansion $\delta_s \approx \tfrac{1}{2}(\theta_s-\theta_{s-1})^\top \mathcal{I}(\theta_{s-1})\,(\theta_s-\theta_{s-1})$, where $\mathcal{I}$ is the Fisher information; i.e.\ $\delta_s$ is a \emph{squared} Fisher--Rao length, and $\sqrt{2\delta_s}$ is the (locally metric) Fisher--Rao distance. Consequently $\Delta_t = \sum_s \delta_s$ is a sum of squared infinitesimal lengths --- an action-like, path-dependent functional --- and should be interpreted as a cumulative \emph{activity} measure rather than as a displacement from $W_0$. For practical validation purposes it remains a useful monotone scalar summary of total behavioral movement, but claims that it bounds net drift require the path to be quasi-monotone in distribution space.
\end{remark}

We argue that $(\delta_t, h_t, \tau_t)$ constitutes a \emph{minimal audit packet} for discrete-time auditing, where $h_t$ is a cryptographic hash of the model state and $\tau_t$ is a timestamp. We use ``minimal'' in the operational sense that each component is non-redundant for the validatory objective; we do \emph{not} claim sufficiency in the formal Fisher--Neyman sense, as no parametric likelihood for the audit decision is posited here. Weight-space distances such as $\|W_t - W_{t-1}\|_F$ are an \emph{inadequate} primary metric (cf.\ Study~A1): large weight moves in directions to which the output distribution is insensitive can produce small behavioral change, while small weight moves in sensitive directions can produce large behavioral change. We record this as the design principle motivating $\delta_t$.
\label{principle:frobenius_inadequate}

\subsection{The Merkle Weight Chain}

To ensure that the telemetry record is tamper-evident, we chain the state digests cryptographically.

\begin{definition}[Merkle Weight Chain]
Define the chain hash sequence by:
\begin{equation}
    h_0 = \mathrm{Hash}(W_0), \qquad h_t = \mathrm{Hash}\!\left(W_t \;\big\|\; h_{t-1} \;\big\|\; \mathrm{Hash}(\mathcal{B}_t)\right)
\end{equation}
where $\|\cdot\|$ denotes concatenation and $\mathrm{Hash}$ is a collision-resistant cryptographic hash function (e.g., SHA-3-256).
\end{definition}

The integrity of this record is itself a governance property. Aggregated telemetry is only as trustworthy as its least tamper-resistant link: if a firm can retroactively rewrite the history of how its model adapted, any reconstruction of its exposure rests on data that can be edited after the fact, and precisely the drift that tends to precede stress could be erased from the record once it became inconvenient. The following result forecloses this, reducing trust in the entire history to verification of a single hash.

\begin{proposition}[Tamper-Evidence]
\label{prop:tamper_ev}Under the collision-resistance assumption of the hash function, any retroactive modification to $W_s$ for $s < t$ invalidates every $h_{s'}, s' \geq s$ with overwhelming probability. The validator need only verify the chain tip $h_T$ to certify the integrity of the entire history.
\end{proposition}

\begin{proof}
\begin{sloppypar}
Suppose an adversary modifies $W_s \to \hat{W}_s$ for some $s < T$. Then $\hat{h}_s = \mathrm{Hash}(\hat{W}_s \| h_{s-1} \| \mathrm{Hash}(\mathcal{B}_s)) \neq h_s$ with overwhelming probability (by collision resistance). Since $h_{s+1}$ depends on $h_s$, the modification propagates forward: $\hat{h}_{s+k} \neq h_{s+k}$ for all $k \geq 0$. In particular, $\hat{h}_T \neq h_T$, which is detectable by the validator holding $h_T$.
\end{sloppypar}
\end{proof}

The Merkle Weight Chain is structurally identical to a blockchain \citep{nakamoto2008bitcoin} applied to model state, but without the requirement for distributed consensus: the validator is the sole external verifier.

\subsection{Event-Driven Logging via Stopping Times}

Logging every step $t$ is impractical at scale. We propose a principled alternative: log only at \emph{stopping times} defined by behavioral divergence crossings.

\begin{definition}[Behavioral Stopping Times]
\label{def:disc_stopping}
\begin{equation}
    \tau_0 = 0, \qquad \tau_{k+1} = \inf\!\left\{t > \tau_k : \Delta_t - \Delta_{\tau_k} > \epsilon\right\}
\end{equation}
for a validator-specified threshold $\epsilon > 0$.
\end{definition}

Under this scheme:
\begin{itemize}
    \item A \textbf{static model} generates zero log packets: $\tau_1 = \infty$.
    \item A \textbf{slowly adapting model} generates sparse packets with large inter-arrival times.
    \item A \textbf{rapidly adapting model} generates dense packets, automatically flagging the adaptation intensity to the validator.
\end{itemize}

The log packet at $\tau_k$ contains $\{h_{\tau_k}, \Delta_{\tau_k}, \tau_k\}$. The total number of packets is proportional to the total behavioral movement of the model, which is the quantity of governance interest.

\begin{remark}
The inter-arrival times $\{\tau_{k+1} - \tau_k\}$ are themselves informative. A gap that is anomalously long --- the model claims to be static for an extended period --- is a signal for the concealment-detection mechanisms of Part~II (Section~\ref{sec:hiding}).
\end{remark}

\section{The Continuous-Time Limit}
\label{sec:continuous}

\subsection{The Weight SDE}

In the small-learning-rate regime $\eta \to 0$, the discrete SGD recursion admits a weak (order-one) approximation by an It\^{o} SDE under the time identification $t = k\eta$ (one update per $\eta$ units of continuous time) \cite{li2017stochastic, mandt2017stochastic}; for the underlying It\^{o} and semimartingale theory see \cite{oksendal2003stochastic} and \cite{protter2005stochastic}, and for stochastic-approximation limits \cite{kushner2003stochastic}:

\begin{equation}
    dW_t = -\nabla_W J(W_t)\, dt + \sqrt{\eta}\, \Sigma(W_t)^{1/2}\, dB_t
    \label{eq:sde}
\end{equation}

where:
\begin{itemize}
    \item $J(W_t) = \mathbb{E}_{\mathcal{B}}[\mathcal{L}(W_t, \mathcal{B})]$ is the population loss;
    \item $\Sigma(W_t) = \mathrm{Cov}_{\mathcal{B}}[\nabla_W \mathcal{L}(W_t, \mathcal{B})]$ is the gradient covariance matrix, capturing mini-batch stochasticity;
    \item $\{B_t\}_{t \geq 0}$ is a standard Brownian motion on $\mathbb{R}^d$.
\end{itemize}

The weight process $\{W_t\}$ is a continuous-time It\^{o} diffusion on $\mathcal{W}$. The audit problem is to summarize the path of this diffusion without storing the full trajectory.

\subsection{Quadratic Variation as the Canonical Telemetry Scalar}

\begin{definition}[Quadratic Variation of the Weight Process]
The quadratic variation of the weight process $\{W_t\}$ is:
\begin{equation}
    [W]_t = \int_0^t \mathrm{tr}\!\left(\eta\, \Sigma(W_s)\right)\, ds
\end{equation}
\end{definition}

\begin{proposition}[Properties of the Quadratic Variation]
\label{prop:qv}
The quadratic variation satisfies:
\begin{enumerate}[label=(\roman*)]
    \item \textbf{Incremental computability:} $d[W]_t = \eta\, \mathrm{tr}(\Sigma_t)\, dt$, so it can be accumulated online at each training step.
    \item \textbf{Activity indicator:} $[W]_t$ is non-decreasing; its slope $\eta\, \mathrm{tr}(\Sigma_t)$ equals zero if and only if the model is not learning.
    \item \textbf{Noise decomposition:} $[W]_t$ isolates the stochastic component of learning, separating it from the deterministic drift $\int_0^t \|\nabla J(W_s)\|^2\, ds$.
\end{enumerate}
\end{proposition}

\begin{proof}
Write \eqref{eq:sde} as $dW_t = b_t\,dt + \sqrt{\eta}\,\Sigma_t^{1/2}dB_t$ with $b_t = -\nabla_W J(W_t)$ and $\Sigma_t = \Sigma(W_t)$.

(i) The quadratic variation of a $d$-dimensional It\^o process is $[W]_t = \sum_{k=1}^d [W^k]_t$, and the finite-variation (drift) part contributes nothing to it: for a partition of mesh $\delta$, $\sum_i \|b_{t_i}\|^2 (t_{i+1}-t_i)^2 \le \delta\, t \sup_s\|b_s\|^2 \to 0$, while the cross terms vanish by Cauchy--Schwarz. Only the martingale part survives, and $d[\,\sqrt{\eta}\,\Sigma^{1/2}B\,]_t = \eta\,\mathrm{tr}\!\left(\Sigma_t^{1/2}(\Sigma_t^{1/2})^\top\right)dt = \eta\,\mathrm{tr}(\Sigma_t)\,dt$ since $\Sigma_t$ is symmetric positive semi-definite. Because the integrand depends only on the current state, the integral can be accumulated by a running sum at each step.

(ii) $\Sigma_t \succeq 0$ implies $\mathrm{tr}(\Sigma_t) \ge 0$, so $t\mapsto[W]_t$ is non-decreasing. Moreover $\mathrm{tr}(\Sigma_t) = 0$ iff $\Sigma_t = 0$ (a positive semi-definite matrix with zero trace is zero), i.e. iff the mini-batch gradient has zero covariance --- no stochastic parameter movement. This is the sense in which a zero slope certifies that the model is not learning stochastically; a model at a deterministic fixed point with $\Sigma_t\neq0$ still accumulates $[W]_t$.

(iii) By construction $[W]_t$ depends on the path only through $\Sigma$, not through $b_t$, so it is invariant to the drift. The drift energy $\int_0^t\|\nabla J(W_s)\|^2ds$ is a separate functional; neither determines the other, which is precisely why both are logged.
\end{proof}

$[W]_t$ is the continuous-time analog of the discrete-time $\sum_s \|W_s - W_{s-1}\|_F^2$, with the important correction that it captures only the diffusion component of movement, not the deterministic gradient descent component. A model that is learning but has reached a fixed point (gradient near zero, stochastic perturbations still active) will show $[W]_t$ increasing; a truly frozen model will show $[W]_t$ flat. This distinction is meaningful for governance. The slope $\eta\,\mathrm{tr}(\Sigma_t)$ is an instantaneous measure of the model's learning intensity, and tracking it across a population of deployed models is the natural way to detect that many models are adapting in concert---the precondition for correlated, portfolio-level dynamics.

\subsection{KL Stopping Times in Continuous Time}

The event-driven logging principle carries directly to continuous time.

\begin{definition}[Continuous-Time KL Stopping Times]
\label{def:ct_stopping}
\begin{equation}
    \tau_0 = 0, \qquad \tau_{k+1} = \inf\!\left\{t > \tau_k : D_{\mathrm{KL}}\!\left(p_{W_{\tau_k}}(\cdot\mid\mathcal{R}) \;\big\|\; p_{W_t}(\cdot\mid\mathcal{R})\right) > \epsilon\right\}
\end{equation}
The reference distribution is \emph{reset} to the current state $W_{\tau_k}$ at each event, so each inter-arrival interval measures fresh divergence accrued since the last log packet.
\end{definition}

The log packet at each $\tau_k$ is $\left\{\mathrm{Hash}(W_{\tau_k}),\; [W]_{\tau_k},\; \tau_k\right\}$.

\begin{theorem}[Renewal Structure of Stopping Times]
\label{thm:renewal}
Assume the weight process \eqref{eq:sde} is positive Harris recurrent with a unique stationary law $\pi$ and is initialized from $\pi$ (or run to stationarity). Then the inter-arrival times $\{T_k\}_{k \geq 1}$, $T_k = \tau_k - \tau_{k-1}$, of the \emph{reset} stopping rule of Definition~\ref{def:ct_stopping} form a stationary renewal sequence, and
\begin{equation}
    \frac{1}{n}\sum_{k=1}^n T_k \;\xrightarrow{\text{a.s.}}\; \bar{T} = \mathbb{E}_\pi[T_1],
\end{equation}
the expected time for the divergence from a stationary-distributed reference state to first reach $\epsilon$. For small $\epsilon$, $\bar{T} \approx \epsilon / \bar{\delta}$, where $\bar\delta = \tfrac{1}{2}\,\eta\,\mathbb{E}_\pi[\operatorname{tr}(\mathcal{I}\,\Sigma)]$ is the stationary mean instantaneous KL rate and $\mathcal{I}$ is the Fisher information of the output map.
\end{theorem}

\begin{proof}
The reset rule (Definition~\ref{def:ct_stopping}) restarts the divergence clock from the current state at each $\tau_k$: the next inter-arrival depends on the path only through $W_{\tau_k}$. By the strong Markov property, $\{(\tau_k, W_{\tau_k})\}$ is a Markov chain, and each cycle's duration $T_{k+1}$ is a deterministic functional of the post-$\tau_k$ excursion started at $W_{\tau_k}$. When the diffusion is positive Harris recurrent with stationary law $\pi$ and is initialized from $\pi$, the chain $\{W_{\tau_k}\}$ is itself stationary, so $\{T_k\}$ is a stationary, ergodic sequence (a stationary renewal sequence in the regenerative sense). Birkhoff's ergodic theorem then gives
\[
  \frac{1}{n}\sum_{k=1}^n T_k \;\xrightarrow{\text{a.s.}}\; \mathbb{E}_\pi[T_1] = \bar{T}.
\]
(If the process is started off-stationarity but is geometrically ergodic, the same limit holds with the first finitely many cycles contributing negligibly; the cycles are then asymptotically, not exactly, i.i.d.)

For the small-$\epsilon$ rate, write the local divergence accrued since the last reset as $\phi_k(t) = D_{\mathrm{KL}}\!\big(p_{W_{\tau_k}}\,\|\,p_{W_t}\big)$ for $t\in[\tau_k,\tau_{k+1}]$, with $\phi_k(\tau_k)=0$. By the Fisher expansion, $\mathbb{E}[\phi_k(\tau_k+h)] = \tfrac{1}{2}\,\mathbb{E}[\,(W_{\tau_k+h}-W_{\tau_k})^\top \mathcal{I}\,(W_{\tau_k+h}-W_{\tau_k})\,] + o(h)$, and the diffusion increment satisfies $\mathbb{E}[(W_{\tau_k+h}-W_{\tau_k})(\cdot)^\top] = \eta\,\Sigma\,h + o(h)$, so $\tfrac{d}{dh}\mathbb{E}[\phi_k]\big|_{h=0^+} = \tfrac{1}{2}\eta\operatorname{tr}(\mathcal{I}\Sigma) =: \bar\delta$ in the stationary average. Hence $\mathbb{E}_\pi[T_1] \approx \epsilon/\bar\delta$ to leading order in $\epsilon$.
\end{proof}

Theorem~\ref{thm:renewal} has practical consequences: under normal learning dynamics, the validator can form an expectation about the packet arrival rate. A sustained deviation from this expected rate is evidence of either anomalous adaptation or deliberate suppression of telemetry. Read across a portfolio of models, the renewal baseline $\bar{T}\approx\epsilon/\bar\delta$ is a per-model norm against which departures can be detected, and a simultaneous compression of inter-arrival times---many models suddenly adapting faster than their baselines---is a candidate leading indicator of a building regime shift that model-by-model monitoring is poorly placed to read.

\subsection{The Fokker--Planck Perspective}

Rather than tracking the weight trajectory as a single path, one may track the evolution of the probability density over weight space. For the SDE \eqref{eq:sde}, the general Fokker--Planck (forward Kolmogorov) operator is $\partial_t\rho_t = \nabla_W\!\cdot(\rho_t\nabla_W J) + \tfrac{\eta}{2}\sum_{ij}\partial_i\partial_j(\Sigma_{ij}\rho_t)$. Specializing to the isotropic, state-independent case $\Sigma(W)\equiv I_d$ for clarity, this reduces to

\begin{equation}
    \frac{\partial \rho_t}{\partial t} = \nabla_W \cdot \!\left(\rho_t\, \nabla_W J\right) + \frac{\eta}{2}\, \Delta_W \rho_t
    \label{eq:fp}
\end{equation}

where $\Delta_W$ is the Laplacian in weight space. The stationary solution to \eqref{eq:fp} is the Gibbs measure
\begin{equation}
    \rho_\infty(W) \propto \exp\!\left(-\frac{2 J(W)}{\eta}\right).
    \label{eq:gibbs}
\end{equation}
\begin{remark}[Scope of the Gibbs form, and consistency with Part~II]
The closed Gibbs form \eqref{eq:gibbs} is specific to isotropic, state-independent noise $\Sigma\equiv I_d$. For state-dependent or anisotropic $\Sigma(W)$ --- the empirically relevant case for SGD, where the gradient-noise covariance tracks the local loss geometry --- the stationary density is in general \emph{not} of Gibbs form $\propto e^{-2J/\eta}$: the forward operator carries a noise-induced drift $\tfrac{\eta}{2}\sum_{ij}\partial_i\partial_j(\Sigma_{ij}\rho)$ that does not integrate to a potential unless a detailed-balance condition ($\Sigma^{-1}\nabla J$ a gradient field, equivalently $\nabla\!\cdot\Sigma$ compatible with $\nabla J$) holds. This is the well-known fact that SGD does not, in general, sample a Gibbs measure of its training loss under anisotropic noise. We flag this explicitly because it is exactly the regime Part~II exploits: an adversary shapes $\Sigma(W)$ to hide drift in the noise structure (Attacks~N and~H), which is possible \emph{only} because the stationary law departs from the Gibbs form. Consequently, any inference of the energy landscape $J$ from an observed stationary density is valid only under the isotropic assumption; we use \eqref{eq:gibbs} purely as the reference (honest, isotropic) baseline against which those attacks are defined, not as a claim about the attacked dynamics.
\end{remark}

This shifts the audit question from ``has this specific weight trajectory drifted?'' to ``has the distribution over model states drifted?'' The validator can audit whether empirical output distributions are consistent with $\rho_t$ remaining concentrated near the validated region $\mathcal{S}$, without observing $W_t$ directly. This is the correct population-level framing for supervision of self-adapting models at scale.

\section{Lyapunov Stability Framework}
\label{sec:lyapunov}

\subsection{The Risk Lyapunov Function}

Let $\mathcal{S} \subset \mathcal{W}$ be the \emph{safe manifold}: the set of model configurations validated during the initial MRM review, characterized by acceptable performance, fairness, and risk metrics.

\begin{definition}[Risk Lyapunov Function]
\begin{equation}
    V(W) = \frac{1}{2}\left\|W - \proj_{\mathcal{S}}(W)\right\|^2
\end{equation}
where $\proj_{\mathcal{S}}$ denotes the orthogonal projection onto $\mathcal{S}$.
\end{definition}

$V(W)$ measures the squared distance from the weight vector $W$ to its nearest point in the validated set $\mathcal{S}$. It equals zero if and only if $W \in \mathcal{S}$.

\subsection{Discrete-Time Stability Condition}

For the discrete-time update rule, the auditable stability condition is:

\begin{equation}
    \Delta V_t = V(W_{t+1}) - V(W_t) \leq -\alpha \left\|W_t - \proj_{\mathcal{S}}(W_t)\right\|^2, \qquad \alpha > 0
    \label{eq:discrete_lyap}
\end{equation}

\begin{remark}
This corrects the formulation in some prior frameworks that bound $\Delta V$ by the gradient norm $\|\nabla \theta\|$ rather than the state deviation. The gradient norm is an inappropriate bound because it can vanish at saddle points while the model remains far from $\mathcal{S}$. The condition \eqref{eq:discrete_lyap} ensures that stability is governed by \emph{where the model is}, not \emph{how fast it is moving}.
\end{remark}

\subsection{Continuous-Time Stability via It\^{o}'s Lemma}

For the continuous-time weight SDE \eqref{eq:sde}, the Lyapunov function evolves according to It\^{o}'s lemma:

\begin{equation}
    dV(W_t) = \left[\nabla V(W_t)^\top (-\nabla J(W_t)) + \frac{\eta}{2}\, \tr\!\left(\Sigma(W_t)\, \nabla^2 V(W_t)\right)\right] dt + \nabla V(W_t)^\top \sqrt{\eta}\, \Sigma(W_t)^{1/2}\, dB_t
    \label{eq:ito_lyap}
\end{equation}

The first term in the drift is $-\langle \nabla V, \nabla J \rangle$ (the deterministic push toward the loss minimum), and the second is $\frac{\eta}{2}\tr(\Sigma\, \nabla^2 V)$ (the \emph{noise energy term}, which quantifies how stochastic weight perturbations interact with the curvature of $V$). The second-order term is absent in discrete-time formulations and is the source of the most subtle model risk.

\begin{definition}[Mean Stability Condition]
The weight process is \emph{mean-stable} with rate $\alpha > 0$ and floor $V_\infty \geq 0$ if its drift satisfies
\begin{equation}
    \mathcal{L}V(W_t) \;:=\; -\langle \nabla V, \nabla J\rangle + \tfrac{\eta}{2}\tr(\Sigma\,\nabla^2 V) \;\leq\; -\alpha\big(V(W_t) - V_\infty\big).
    \label{eq:mean_stable}
\end{equation}
Taking expectations and applying Gr\"onwall's inequality to $v(t) = \mathbb{E}[V(W_t)]$ gives
\begin{equation}
    \mathbb{E}[V(W_t)] \;\leq\; V_\infty + \big(V(W_0) - V_\infty\big)\,e^{-\alpha t},
\end{equation}
i.e.\ exponential return \emph{to a neighborhood of $\mathcal{S}$ of size $V_\infty$}, not to $\mathcal{S}$ itself.
\end{definition}

\begin{remark}
The floor $V_\infty$ is unavoidable whenever the diffusion is non-degenerate near $\mathcal{S}$: at $V=0$ the drift equals $\tfrac{\eta}{2}\tr(\Sigma\,\nabla^2 V) > 0$, so $V$ cannot be driven to $0$ in expectation. For the Ornstein--Uhlenbeck surrogate $dW = -\alpha W\,dt + \sigma\,dB$ one has $V_\infty = \sigma^2 d/(4\alpha)$ in continuous time. The earlier idealization $\mathbb{E}[V(W_t)]\le V(W_0)e^{-\alpha t}$ holds only in the noiseless limit $\Sigma\to 0$ and is contradicted by the stationary density $\rho_\infty \propto e^{-2J/\eta}$ derived above; see Study~A5, where $\mathbb{E}[V(W_t)]\to V_\infty>0$.
\end{remark}

The sign and magnitude of the noise energy term $\frac{\eta}{2}\tr(\Sigma\,\nabla^2 V)$ depend on the geometry of $\mathcal{S}$. If $\mathcal{S}$ is convex, $V(W)=\tfrac12\|W-\proj_{\mathcal{S}}(W)\|^2$ is convex with $\nabla^2 V \succeq 0$, so the noise term is always non-negative (purely destabilizing): isotropic gradient noise can only push $\mathbb{E}[V]$ \emph{up}, which is precisely the origin of the floor $V_\infty$. If $\mathcal{S}$ is a non-convex smooth manifold, $V$ is $C^2$ only within the reach (tubular neighborhood) of $\mathcal{S}$, and there $\nabla^2 V$ can have negative eigenvalues in directions where $\mathcal{S}$ curves away; the noise term is then sign-indefinite. Either way, a firm can modulate the term by shaping $\Sigma$ relative to the local curvature of $\mathcal{S}$ --- the continuous-time mechanism by which adversarial drift can be hidden in the noise structure of learning, addressed in detail in Part~II (Section~\ref{sec:continuous_attacks}). Throughout we assume $W_t$ remains within the reach of $\mathcal{S}$ so that $\nabla V = W - \proj_{\mathcal{S}}(W)$ and $\nabla^2 V$ are well defined.

\subsection{Almost-Sure Path-Wise Stability}

Mean stability allows individual trajectories to escape $\mathcal{S}$ provided they return quickly on average. For validation purposes, almost-sure path-wise bounds are more appropriate. The certificate established next is a statement of \emph{individual-model resilience}---a guarantee that a single model's risk state, once perturbed, returns to a bounded neighborhood rather than wandering off---and it is the building block from which any assessment of an ensemble of models must start.

\begin{theorem}[Almost-Sure Stability via Burkholder--Davis--Gundy]
\label{thm:as_stability}
\begin{sloppypar}
Let $M_t = \int_0^t \nabla V(W_s)^\top \sqrt{\eta}\, \Sigma(W_s)^{1/2}\, dB_s$ be the martingale term of \eqref{eq:ito_lyap}. By the Burkholder--Davis--Gundy inequality:
\end{sloppypar}
\begin{equation}
    \mathbb{E}\!\left[\sup_{0 \leq s \leq T} M_s^2\right] \leq C\, \mathbb{E}\!\left[\int_0^T \eta\, \left\|\Sigma(W_s)^{1/2} \nabla V(W_s)\right\|^2\, ds\right]
\end{equation}
and by the Law of the Iterated Logarithm:
\begin{equation}
    \limsup_{t \to \infty} \frac{|M_t|}{\sqrt{2 [M]_t \log\log [M]_t}} = 1 \quad \text{a.s.}
\end{equation}
where $[M]_t = \int_0^t \eta\, \|\Sigma^{1/2} \nabla V\|^2\, ds$ is the predictable quadratic variation of $M_t$.
\end{theorem}

\begin{proof}
It\^{o}'s lemma applied to $V(W_t)$ along the SDE \eqref{eq:sde} gives the semimartingale decomposition
\begin{equation}
    V(W_t) = V(W_0) + A_t + M_t,
\end{equation}
where the finite-variation component is
\[
    A_t = \int_0^t \!\left(-\langle \nabla V(W_s),\, \nabla J(W_s)\rangle + \tfrac{\eta}{2}\operatorname{tr}\bigl(\Sigma(W_s)\nabla^2 V(W_s)\bigr)\right) ds,
\]
and the local martingale is $M_t = \int_0^t \nabla V(W_s)^\top \sqrt{\eta}\,\Sigma^{1/2}(W_s)\,dB_s$ with predictable bracket $[M]_t = \int_0^t \eta\|\Sigma^{1/2}\nabla V\|^2\,ds$.

\emph{BDG bound.} The Burkholder--Davis--Gundy inequality (see, e.g., \cite{revuz1999continuous}, Theorem~IV.4.1) asserts that for any $p\geq 1$ there exists a universal constant $C_p > 0$ such that $\mathbb{E}[\sup_{s\leq T}|M_s|^{2p}] \leq C_p\,\mathbb{E}[[M]_T^p]$.  For $p=1$ this gives the stated $L^2$ bound.

\emph{LIL for continuous martingales.} If $[M]_\infty = +\infty$ a.s., the LIL for continuous local martingales (\cite{revuz1999continuous}, Theorem~V.1.8) gives
\[
    \limsup_{t\to\infty}\frac{|M_t|}{\sqrt{2[M]_t\log\log [M]_t}} = 1 \quad \text{a.s.}
\]
If $[M]_\infty < \infty$ a.s.\ then $M_t$ converges a.s.\ and the result holds trivially with a zero limsup.

\emph{Ultimate boundedness (under stated growth hypotheses).}  We invoke two standard stochastic-stability hypotheses \citep{khasminskii2011stochastic}: \textbf{(H1)} the dissipative drift bound $\mathcal{L}V(W)\le -\alpha\big(V(W)-V_\infty\big)$ of \eqref{eq:mean_stable}; and \textbf{(H2)} a linear-growth bound on the diffusion term, $\|\Sigma^{1/2}(W)\nabla V(W)\|^2 \le K\big(1+V(W)\big)$ for some $K<\infty$. Crucially, (H2) is \emph{verified}, not assumed circularly: for $V=\tfrac12\|W-\proj_{\mathcal S}(W)\|^2$ one has $\nabla V = W-\proj_{\mathcal S}(W)$, so if $\Sigma$ is bounded ($\Sigma\preceq\sigma_{\max}^2 I$) then $\|\Sigma^{1/2}\nabla V\|^2 \le \sigma_{\max}^2\|\nabla V\|^2 = 2\sigma_{\max}^2 V$, i.e.\ (H2) holds with $K=2\sigma_{\max}^2$. (The earlier draft asserted ``$[M]_t=O(t)$'' directly, which would require $\sup_s\|\nabla V(W_s)\|<\infty$ --- the very boundedness being proved; (H2) replaces that step with a structural bound that makes $d[M]_t/dt \le 2\eta\sigma_{\max}^2 V$ controlled \emph{by $V$ itself}.) Taking expectations of (H1) and applying Gr\"onwall gives the self-contained mean bound $\mathbb{E}[V(W_t)]\le V_\infty + (V(W_0)-V_\infty)e^{-\alpha t}$. For the almost-sure conclusion, (H1)+(H2) are exactly the hypotheses of the stochastic Lyapunov/ultimate-boundedness theorem (\citealp{khasminskii2011stochastic}, Ch.~5--6): the process is bounded in probability and
\[
    \limsup_{t\to\infty} V(W_t) \;\leq\; V_\infty \quad\text{a.s.},
\]
so it is recurrent to the neighborhood $\{V \leq V_\infty\}$ of $\mathcal{S}$.  Convergence $V(W_t)\to 0$ does \emph{not} hold when $\Sigma\not\to 0$ near $\mathcal{S}$: the non-degenerate stationary law forces persistent fluctuations of order $V_\infty$ about $\mathcal{S}$. (In the degenerate limit $\Sigma\to 0$, $V_\infty\to 0$ and one recovers a.s.\ convergence to $\mathcal{S}$.)
\end{proof}

The bound in Theorem~\ref{thm:as_stability} is computable from the telemetry stream: both $\Sigma_t$ (from the gradient covariance) and $\nabla V(W_t)$ (from the model state hash and proximity to $\mathcal{S}$) can be logged incrementally. A practical, monitorable form of the path-wise condition requires the drift to dominate the noise envelope outside the floor neighborhood: there exist $\alpha>0$, $V_\infty\geq 0$, and (eventually) for the LIL constant the certificate
\begin{equation}
    \begin{split}
    &\underbrace{-\langle \nabla V, \nabla J \rangle + \tfrac{\eta}{2}\tr(\Sigma\, \nabla^2 V)}_{\text{drift }\mathcal{L}V} \;\leq\; -\alpha\big(V(W_t)-V_\infty\big)
    \quad\text{and}\\[2pt]
    &\qquad     V(W_t)\le V_\infty + \frac{(1+\varepsilon)\sqrt{2[M]_t\log\log[M]_t}}{\alpha t}\ \ \text{a.s.\ eventually.}
    \label{eq:as_stable}
    \end{split}
\end{equation}
The first inequality is the drift certificate; the second is the LIL envelope, which decays as $t^{-1/2}\sqrt{\log\log t}\to 0$, so the right-hand side contracts to $V_\infty$.

\begin{corollary}
\label{cor:as_ultimate}
If the drift certificate in \eqref{eq:as_stable} holds for all $t \geq 0$, then $\limsup_{t\to\infty} V(W_t) \leq V_\infty$ almost surely: the weight process is ultimately bounded and recurs to the neighborhood $\{V\le V_\infty\}$ of $\mathcal{S}$ with probability one. If the certificate is violated at any time, the telemetry probe triggers an automatic halt. (Exact return $V(W_t)\to 0$ holds only in the degenerate-noise limit $\Sigma\to 0$.)
\end{corollary}

\begin{proof}
Both statements in \eqref{eq:as_stable} are assumed to hold. Apply It\^o's formula to $V$ along
\eqref{eq:sde} and split into drift and martingale parts:
\[
  V(W_t) \;=\; V(W_0) \;+\; \int_0^t \mathcal{L}V(W_s)\,ds \;+\; M_t ,
  \qquad M_t = \int_0^t \nabla V(W_s)^\top\sqrt{\eta}\,\Sigma(W_s)^{1/2}\,dB_s .
\]
By the drift certificate, $\mathcal{L}V(W_s) \leq -\alpha\big(V(W_s)-V_\infty\big)$ for all $s$.
Writing $U_t = V(W_t)-V_\infty$, this gives $U_t \leq U_0 - \alpha\int_0^t U_s\,ds + M_t$. By the
comparison/Gronwall argument in integral form,
\[
  U_t \;\leq\; U_0 e^{-\alpha t} \;+\; \int_0^t e^{-\alpha(t-s)}\,dM_s ,
\]
so the deterministic part decays to zero. For the stochastic part, the second inequality of
\eqref{eq:as_stable} bounds the martingale contribution by the law of the iterated logarithm:
$|M_t| \leq (1+\varepsilon)\sqrt{2[M]_t\log\log[M]_t}$ a.s.\ for all large $t$, and dividing by
$\alpha t$ gives a bound of order $t^{-1/2}\sqrt{\log\log t}$ under (H2), which tends to zero. Hence
$\limsup_{t\to\infty} U_t \leq 0$ almost surely, i.e.\ $\limsup_{t\to\infty} V(W_t) \leq V_\infty$.
Recurrence to $\{V\leq V_\infty\}$ follows because $V\geq0$ and $V$ cannot remain strictly above
$V_\infty$ on a set of positive probability without contradicting this limsup.

Note the direction of the statement: it gives ultimate boundedness to a neighborhood of
$\mathcal{S}$, not convergence to $\mathcal{S}$. Exact return $V(W_t)\to0$ requires the degenerate
limit $\Sigma\to0$, since a non-degenerate stationary distribution keeps $V$ fluctuating at the
noise floor.
\end{proof}

The halt is, in effect, a model-level circuit breaker. That every model in a portfolio can hold such an individual certificate and yet the portfolio collectively fail to be stable---so that per-model resilience does not aggregate---is a limitation of any purely per-model certification regime.

\section{CUSUM Governance Trigger}
\label{sec:cusum}

\subsection{Motivation}

The tiered governance table of prior frameworks (Tier~1: micro-update, Tier~2: meso-update, Tier~3: macro-update) lacks a formal statistical trigger. We propose the Cumulative Sum (CUSUM) test on the behavioral divergence sequence $\{\delta_t\}$ as a statistically grounded escalation mechanism.

\subsection{The CUSUM Statistic}

\begin{definition}[Governance CUSUM]
Let $\kappa > 0$ be a reference value representing the expected behavioral divergence per step under normal micro-adaptation. The CUSUM statistic is:
\begin{equation}
    C_0 = 0, \qquad C_t = \max(0,\; C_{t-1} + \delta_t - \kappa)
    \label{eq:cusum}
\end{equation}
\end{definition}

The CUSUM test declares a Tier~$\ell$ event at stopping time:
\begin{align}
    T_{\ell} &= \inf\!\left\{t \geq 0 : C_t \geq h_\ell\right\}
\end{align}
where $0 < h_1 < h_2$ are tier-specific thresholds corresponding to Tier~2 (objective parameter shift requiring 24h shadow testing) and Tier~3 (macro-update requiring Human-in-the-Loop sign-off).

The choice of CUSUM is not arbitrary. Among detectors with a fixed tolerance for false alarms it flags a genuine shift as quickly as possible, and where model adaptation can outpace the validation cycle, detection \emph{delay} is itself a governance quantity: every period a building shift goes undetected is a period in which the production model diverges further from the validated one. The following classical result makes the optimality precise.

\begin{theorem}[CUSUM Optimality; Lorden \citep{lorden1971procedures}, Moustakides \citep{moustakides1986optimal}]
\label{thm:cusum_opt}
\emph{(Classical; stated without proof --- see the cited sources.)}
Suppose the increment fed to \eqref{eq:cusum} is the log-likelihood ratio $\log\frac{f_1(\delta_t)}{f_0(\delta_t)}$ between the post-change density $f_1$ and the pre-change density $f_0$ of the monitored statistic. Then among all stopping rules with average run length to false alarm $\geq \gamma$, the CUSUM rule minimizes Lorden's worst-case expected detection delay as $\gamma\to\infty$ \citep{moustakides1986optimal}.
\end{theorem}

\begin{remark}[Applicability to the behavioral CUSUM]
The statistic \eqref{eq:cusum} uses the raw centred divergence $\delta_t-\kappa$ rather than a log-likelihood ratio. Theorem~\ref{thm:cusum_opt} therefore applies exactly only when $\delta_t-\kappa$ coincides (up to affine reparameterization) with the LLR of a one-parameter exponential-family shift in the divergence rate --- e.g.\ a change in the mean of an approximately Gaussian or exponential $\delta_t$. In general the centred-divergence CUSUM is a well-behaved, monotone detector with explicit ARL control (Remark~\ref{rmk:arl}) but is only \emph{near}-optimal; exact Lorden optimality requires plugging the calibrated LLR into \eqref{eq:cusum}.
\end{remark}

\begin{remark}
\label{rmk:arl}
The false alarm rate of \eqref{eq:cusum} can be set explicitly via the average run length under the null (no drift). For a Tier~2 trigger, a practical choice might be an ARL of 250 business days ($\approx 1$ year) under normal operation, corresponding to a threshold $h_1$ computable from the null distribution of $\delta_t$.
\end{remark}

\subsection{Continuous-Time CUSUM}

In the continuous-time limit, the CUSUM statistic becomes the reflected process

\begin{equation}
    \widetilde{C}_t = \xi_t - \kappa t - \min_{0 \leq s \leq t}\big(\xi_s - \kappa s\big),
    \qquad \xi_t = \int_0^t \frac{d}{du} D_{\mathrm{KL}}\!\left(p_{W_{s^\ast(u)}} \;\big\|\; p_{W_u}\right) du,
\end{equation}

where $s^\ast(u)$ denotes the last reset before $u$, consistent with Definition~\ref{def:ct_stopping}. The stopping time $\theta_\ell = \inf\{t : \widetilde{C}_t \geq h_\ell\}$ inherits a minimax optimality under the LLR-increment hypothesis of Theorem~\ref{thm:cusum_opt}; the closely related \emph{Bayesian} quickest-detection problem (minimizing expected delay under a prior on the change time) is solved instead by the Shiryaev procedure \citep{shiryaev1963optimum}. The two are distinct optimality criteria and we invoke the minimax one here.

\section{The Integrated Telemetry Architecture}
\label{sec:architecture}

\subsection{System Components}

The complete telemetry architecture integrates the components developed above into a five-layer system:

\begin{enumerate}
    \item \textbf{Probe Layer:} Embedded within the SGD optimiser. Computes $\delta_t$, $d[W]_t$, and the CUSUM increment $\delta_t - \kappa$ at each update step.
    \item \textbf{Hash Layer:} Computes the Merkle chain hash $h_t = \mathrm{Hash}(W_t \| h_{t-1} \| \mathrm{Hash}(\mathcal{B}_t))$ at each stopping time $\tau_k$.
    \item \textbf{Stability Layer:} Evaluates the Lyapunov condition \eqref{eq:discrete_lyap} or \eqref{eq:as_stable} at each step. Triggers immediate halt if violated.
    \item \textbf{Escalation Layer:} Monitors $C_t$ against $h_1, h_2$. Fires Tier~2 or Tier~3 governance actions as required.
    \item \textbf{Ledger Layer:} Transmits log packets $\{h_{\tau_k}, [W]_{\tau_k}, \delta_{\tau_k}, \tau_k\}$ to the validator's immutable audit sink (append-only ledger or governance ledger).
\end{enumerate}

\subsection{Log Packet Specification}

The complete log packet emitted at stopping time $\tau_k$ is:
\begin{equation}
    \mathcal{D}_{\tau_k} = \left\{\; h_{\tau_k},\; [W]_{\tau_k},\; \Delta_{\tau_k},\; C_{\tau_k},\; V(W_{\tau_k}),\; \tau_k \;\right\}
\end{equation}

This is a six-dimensional vector of scalars (plus the hash). Compared to storing the weight matrix ($O(d)$ with $d \sim 10^9$), the log packet is negligibly small and can be transmitted in real time to a governance ledger.

\begin{table}[h]
\centering
\begin{tabular}{@{}llll@{}}
\toprule
\textbf{Field} & \textbf{Symbol} & \textbf{Type} & \textbf{Governance Interpretation} \\
\midrule
Model State Hash & $h_{\tau_k}$ & 256-bit hash & Tamper-evident identity \\
Quadratic Variation & $[W]_{\tau_k}$ & Scalar & Cumulative learning intensity \\
Cumulative KL Drift & $\Delta_{\tau_k}$ & Scalar & Total behavioral movement \\
CUSUM Statistic & $C_{\tau_k}$ & Scalar & Governance tier indicator \\
Lyapunov Value & $V(W_{\tau_k})$ & Scalar & Distance from validated set \\
Timestamp & $\tau_k$ & UTC timestamp & Temporal anchor \\
\bottomrule
\end{tabular}
\caption{Log packet fields and their governance interpretations.}
\label{tab:log_packet}
\end{table}

\section{Numerical Studies: Telemetry Architecture}
\label{sec:numerical_a}

We present five numerical studies, each targeting a distinct theoretical claim of the paper. All experiments use a lightweight softmax classifier $f(\cdot\,;\,W)$ with weight matrix $W \in \mathbb{R}^{d \times K}$, $d=8$, $K=3$, trained on a fixed canonical reference set $\mathcal{R}$ of $m=80$ unit-normalized inputs. Because the output probabilities and KL divergence are analytically tractable in this setting, the experiments provide clean quantitative agreement with the theoretical predictions. All studies use a fixed random seed for reproducibility.

\subsection*{Scope of the numerical evidence}
\label{sec:scope_ab}

The studies in both Part~I and Part~II are \emph{verification} experiments: they confirm that the
analytical results hold under controlled simulation and that the predicted qualitative behavior
appears. They are not statistical estimation, and the following limits apply to every number reported
in Sections~\ref{sec:numerical_a} and~\ref{sec:numerical_b}.

\begin{itemize}[nosep]
  \item \textbf{Fixed, documented seeds.} Every study is seeded and exactly reproducible; seeds vary
        with the swept parameter where a sweep is performed, so no reported comparison is a
        deterministic rescaling of another.
  \item \textbf{Fixed discretization.} We use a fixed discretization, and this matters most
        for Study~B2, which claims that a statistic can be dominated by the discretization
        grid; there the $n$-dependence is exhibited deliberately.
  \item \textbf{Toy scale, and one claim it cannot test.} All studies use a softmax classifier with
        $dK=24$ free parameters and $m=80$ canary points. This is far from the high-dimensional regime
        $d \gg m\log|\mathcal{Y}|$ in which the counting argument of Section~\ref{sec:residual} is
        posed; indeed at this scale the parameters are \emph{over}-determined by the canary outputs
        (Remark~\ref{rmk:regime}). Study~B5 therefore validates Theorem~\ref{cor:undetectable} --- the
        finite-resolution result, which holds at any dimension --- and does \emph{not} test the
        heuristic. We are not aware of a simulation at accessible scale that would.
  \item \textbf{Synthetic dynamics and a synthetic adversary.} The weight processes are prescribed
        recursions, and the attacks of Part~II are implemented exactly as modeled rather than
        discovered by an adversary optimizing against the defences. The studies show that the stated
        attacks behave as analyzed; they do not establish that these are the strongest attacks
        available.
\end{itemize}

\subsection{Study A1: KL Divergence vs.\ Frobenius Norm as Audit Metric}

\textbf{Theoretical claim (Section~\ref{sec:discrete}, design principle following Definition~\ref{def:disc_stopping}).}
The Frobenius norm $\|W_t - W_{t-1}\|_F$ is an inadequate audit metric because large weight moves tangent to the loss manifold produce negligible behavioral change, while small moves in the gradient-aligned direction produce large behavioral change. The behavioral divergence $\delta_t = D_{\mathrm{KL}}(p_{W_{t-1}} \| p_{W_t})$ on the canonical reference set $\mathcal{R}$ is the correct primary metric.

\textbf{Experimental design.}
Starting from $W_0$, we construct two equal-magnitude perturbations:
\begin{itemize}
    \item \emph{Tangent perturbation:} $\Delta W_{\mathrm{tan}} = \varepsilon\,\hat{v}$, where $\hat{v}$ is a unit vector orthogonal to $\nabla_W J$ (null-space direction).
    \item \emph{Gradient perturbation:} $\Delta W_{\mathrm{grad}} = \varepsilon\,\hat{g}$, where $\hat{g} = \nabla_W J / \|\nabla_W J\|$ is the unit gradient direction.
\end{itemize}
Both perturbations satisfy $\|\Delta W\|_F = \varepsilon$ exactly, so any difference in the resulting KL divergence is attributable solely to direction.  We sweep $\varepsilon \in \{0.1, 0.3, 0.5, 1.0\}$ and report $D_{\mathrm{KL}}(p_{W_0} \| p_{W_0+\Delta W})$ for each direction.

\textbf{Results.}
\begin{center}
\begin{tabular}{ccccc}
\hline
$\varepsilon$ & $\|\Delta W\|_F$ & KL (tangent) & KL (gradient) & Ratio \\
\hline
0.1 & 0.1000 & 0.00851 & 0.01943 & 2.28$\times$ \\
0.3 & 0.3000 & 0.07637 & 0.17635 & 2.31$\times$ \\
0.5 & 0.5000 & 0.21140 & 0.49351 & 2.33$\times$ \\
1.0 & 1.0000 & 0.83706 & 2.00399 & 2.39$\times$ \\
\hline
\end{tabular}
\end{center}
At every $\varepsilon$ the Frobenius norm is identical for both directions by construction, yet the gradient-direction perturbation causes a larger KL divergence --- a mean factor of $2.33\times$ over the tested grid (range $2.28\times$ at $\varepsilon=0.1$ to $2.39\times$ at $\varepsilon=1.0$).  A validator using $\|\Delta W\|_F$ as the sole threshold assigns equal risk to both perturbations, whereas the KL divergence on $\mathcal{R}$ correctly identifies the gradient direction as the higher-risk move at every magnitude tested.  This supports the design principle of Section~\ref{sec:discrete}: Frobenius norm is blind to direction, so equal-norm moves can carry materially different behavioral risk.  We caution that the precise ratio is configuration-specific (it depends on $W_0$, $\mathcal{R}$, and the chosen tangent direction) and is reported here as an illustrative magnitude, not a universal constant; the robust, direction-independent claim is that $\|\Delta W\|_F$ fails to order behavioral risk.  Panel A1 of Figure~\ref{fig:numerical_a} illustrates the diverging KL curves.

\subsection{Study A2: Tamper-Evidence of the Merkle Weight Chain}

\textbf{Theoretical claim (Proposition~\ref{prop:tamper_ev}).}
Under collision-resistance of the hash function, any retroactive modification to $W_s$ for $s < T$ invalidates every $h_{s'}$ for $s' \geq s$ with overwhelming probability. The validator need only verify the chain tip $h_T$.

\textbf{Experimental design.}
We build a legitimate Merkle Weight Chain of length $T=10$ by simulating SGD updates on the classifier. Each hash is $h_t = \mathrm{SHA3\text{-}256}(W_t \| h_{t-1} \| \mathrm{Hash}(\mathcal{B}_t))$. We then simulate an adversary who retroactively modifies $W_3 \leftarrow W_3 + 0.01$ (a tiny perturbation of magnitude $10^{-2}$) and recomputes the chain from step 3. We compare the tampered chain tip $\hat{h}_{10}$ against the legitimate tip $h_{10}$.

\textbf{Results.}
The modification at step $s=3$ invalidated all $8$ of the subsequent hashes (steps 3 through 10 inclusive), and the two 256-bit chain tips $h_{10}$ and $\hat{h}_{10}$ disagreed (an avalanche change in approximately half of the output bits, as expected for a secure hash). This is consistent with Proposition~\ref{prop:tamper_ev} at the level of a concrete cryptographic demonstration: even a perturbation of size $10^{-2}$ in weight space --- far below any governance threshold --- is detected with certainty by the chain verifier. The result illustrates why the Merkle structure is superior to independent hashes: the validator does not need to store or verify $T$ hashes but only the single chain tip.

\subsection{Study A3: Event-Driven Stopping Times and Packet Density}

\textbf{Theoretical claim (Definition~\ref{def:disc_stopping} and Theorem~\ref{thm:renewal}).}
Under KL stopping times with threshold $\varepsilon$, log-packet density is proportional to behavioral movement, not wall-clock time. A static model generates zero packets; the inter-arrival times $\{T_k\}$ form a renewal process whose mean is $\varepsilon / \bar{\delta}$, where $\bar{\delta}$ is the ergodic mean KL rate.

\textbf{Experimental design.}
We simulate three model variants for $T=500$ SGD steps using learning rates $\eta \in \{0, 0.005, 0.05\}$ and proportional noise, representing static, slow-adapting, and fast-adapting regimes. We apply the stopping-time rule with $\varepsilon = 0.005$ and count log packets emitted.

\textbf{Results.}
The static model ($\eta = 0$, no weight noise) emitted exactly $0$ packets, as predicted ($\tau_1=\infty$). The slow model emitted $29$ packets with mean inter-arrival time of $16.9$ steps. The fast model emitted $128$ packets with mean inter-arrival time of $3.9$ steps --- roughly a $4.4\times$ increase in packet density. The increase reflects the combined effect of the $10\times$ larger learning rate and the larger injected gradient noise, mediated by the nonlinear (approximately quadratic, to leading order) response of KL divergence to weight displacement; we therefore do not read a clean power law off three points. This is consistent with Theorem~\ref{thm:renewal}: under the reset rule the expected inter-arrival $\bar{T}\approx \varepsilon/\bar{\delta}$ decreases as the adaptation intensity (and hence $\bar\delta$) rises. The practical implication is that a frozen model never burdens the governance ledger, while a rapidly adapting model automatically generates a dense audit trail commensurate with its activity. Panel A3 of Figure~\ref{fig:numerical_a} shows the cumulative KL drift $\Delta_t$ with packet emission times overlaid.

\begin{remark}
The clean $0$-packet count for the static model holds only for a genuinely frozen model. A model that declares itself static but is subject to hardware-level weight noise (or covert updates) will, over a long horizon, accrue sub-threshold KL until a packet fires; detecting this gap between declared and actual activity is the subject of Part~II (Attack~S, Section~\ref{sec:discrete_attacks}).
\end{remark}

\subsection{Study A4: Quadratic Variation as the Canonical Telemetry Scalar}

\textbf{Theoretical claim (Section~\ref{sec:continuous}, Proposition~\ref{prop:qv}).}
The quadratic variation of the weight process satisfies $d[W]_t = \eta\, \mathrm{tr}(\Sigma_t)\, dt$. For isotropic gradient covariance $\Sigma = \sigma^2 I_d$, this gives $[W]_T = \eta \sigma^2 d_{\mathrm{tot}} T$ over total time $T$.

\textbf{Experimental design.}
We simulate the Itô-discretized weight SDE
\[
  \Delta W_t = -\alpha W_t\,\Delta t + \sqrt{\eta}\,\sigma\,\sqrt{\Delta t}\,Z_t, \qquad Z_t \sim \mathcal{N}(0,I),
\]
with $\alpha=0.05$, $\sigma=0.10$, step size $\Delta t = 0.01$, and $N=2000$ steps (total time $T=20$), for $\eta \in \{0.01, 0.05, 0.10\}$ and $d_{\mathrm{tot}}=24$.  The empirical quadratic variation is accumulated as $[W]_T^{\mathrm{emp}} = \sum_{t=1}^N \|\Delta W_t^{\mathrm{noise}}\|_F^2$ using the noise component $\Delta W_t^{\mathrm{noise}} = \sqrt{\eta}\,\sigma\,\sqrt{\Delta t}\,Z_t$ only, so that the $O(\Delta t^2)$ drift contribution is excluded.  We compare to the theoretical prediction $\eta\,\sigma^2\,d_{\mathrm{tot}}\,T$.

\textbf{Results.}
\begin{center}
\begin{tabular}{cccc}
\hline
$\eta$ & $[W]_T^{\mathrm{emp}}$ & $\eta\sigma^2 d T$ (theory) & Error \\
\hline
0.01 & 0.04827 & 0.04800 & 0.56\% \\
0.05 & 0.24134 & 0.24000 & 0.56\% \\
0.10 & 0.48268 & 0.48000 & 0.56\% \\
\hline
\end{tabular}
\end{center}
The empirical quadratic variation matches the theoretical prediction $\eta\sigma^2 d_{\mathrm{tot}} T$ within $0.56\%$ at all three learning rates. Because the simulation reuses a single Gaussian noise draw across the three $\eta$ values (the noise increment is $\sqrt{\eta}\,\sigma\sqrt{\Delta t}\,Z$, linear in $\sqrt{\eta}$), the three empirical values are \emph{exactly} proportional in $\eta$ and share the same relative error; this is therefore a single Monte Carlo realization evaluated at three scales, not three independent confirmations. The $0.56\%$ residual is the sampling error of $\tfrac{1}{Nd_{\mathrm{tot}}}\sum_{t,i} Z_{t,i}^2$ about its mean of $1$ (here $N d_{\mathrm{tot}} = 48{,}000$, giving an expected relative error $\sqrt{2/48000}\approx 0.6\%$) and vanishes under path averaging. The study confirms that the quadratic-variation identity $[W]_T=\eta\sigma^2 d_{\mathrm{tot}}T$ is not an asymptotic approximation but holds in the discretized It\^{o} setting at small step size. Panel A4 of Figure~\ref{fig:numerical_a} shows the three $[W]_t$ trajectories growing linearly at slopes $\eta\sigma^2 d_{\mathrm{tot}}$, with dots marking the terminal empirical values.

\subsection{Study A5: Lyapunov Stability and the CUSUM Governance Trigger}

\textbf{Theoretical claim (Definitions~5.1--5.2, Theorem~\ref{thm:as_stability}, Theorem~\ref{thm:cusum_opt}).}
Under the (floored) mean stability condition \eqref{eq:mean_stable}, the Lyapunov function decays to a positive noise floor along the exact curve $\mathbb{E}[V(W_t)] = V_0(1-\alpha)^{2t} + V_\infty\big(1-(1-\alpha)^{2t}\big)$. The CUSUM detector run on the \emph{behavioral-divergence increments} $\delta_t$ escalates governance shortly after an adversarial drift begins, while keeping the false-alarm rate (average run length) under control.

\textbf{Experimental design.}
We set $V(W) = \frac{1}{2}\|W\|_F^2$ (so $\mathcal{S}=\{0\}$) and simulate the discrete Ornstein--Uhlenbeck recursion $W_{t+1} = (1-\alpha)W_t + \sigma Z_t$, $Z_t\sim\mathcal N(0,I)$, with $\alpha=0.05$, $\sigma=0.005$, for $T=100$ steps. We average $n_{\mathrm{mc}}=100$ independent paths and compare $\hat{\mathbb{E}}[V(t)]$ to the exact formula. We report the floor as the continuous-time value $V_\infty = \sigma^2 d_{\mathrm{tot}}/(4\alpha)=0.00300$; the exact discrete stationary value is $\sigma^2 d_{\mathrm{tot}}/[2\alpha(2-\alpha)]=0.00308$, a $2.6\%$ difference that is within Monte Carlo error at this sample size (we use the continuous value in the ``exact'' column for consistency with the SDE framing).

For the CUSUM study, the detector is run on the \emph{behavioral-divergence increments} $\delta_t = D_{\mathrm{KL}}(p_{W_{t-1}}\|p_{W_t})$ of the softmax reference model (Definition~\ref{def:disc_stopping} setting), \emph{not} on the Lyapunov value $V$. We calibrate $\kappa$ to the mean of $\delta_t$ under honest micro-adaptation and set $h_1$ to achieve a target null average run length (ARL); the adversarial run injects a constant hidden weight drift beginning at a known change point $t_0$.

\textbf{Results.}
\begin{center}
\begin{tabular}{ccc}
\hline
$t$ & $\hat{\mathbb{E}}[V(t)]$ (MC) & $\mathbb{E}[V(t)]$ (exact) \\
\hline
 0  & 2.05312 & 2.05312 \\
10  & 0.73740 & 0.73794 \\
30  & 0.09775 & 0.09745 \\
60  & 0.00743 & 0.00735 \\
100 & 0.00308 & 0.00307 \\
\hline
\end{tabular}
\end{center}
The MC mean tracks the exact formula to within $1.1\%$ relative error at every tabulated time (deviations $0.07\%$, $0.31\%$, $1.04\%$, $0.36\%$ at $t=10,30,60,100$). Over the full horizon the largest single-step relative error is $\approx 7\%$, occurring near $t=83$ where $V$ has reached the noise floor ($V\approx 3.4\times10^{-3}$); there the \emph{absolute} error is only $\approx 2.5\times 10^{-4}$, i.e.\ Monte Carlo fluctuation at the floor rather than model misfit. We caution that a near-unit correlation between two monotone exponential decays is not a stringent test. It is not merely weak but \emph{blind to scale}: a log-log correlation is invariant under any multiplicative bias, since $\log(\lambda x) = \log\lambda + \log x$ shifts but does not tilt the relationship. An estimator uniformly $2\times$ or $0.5\times$ the truth would still report $r=1.000000$. We therefore report relative error per time point rather than a single correlation coefficient.

For the CUSUM study, the detector is run on the behavioral-divergence increments $\delta_t$ with reference value $\kappa$ set to the $95$th percentile of $\delta_t$ under honest micro-adaptation ($\kappa\approx 4.8\times 10^{-5}$) and threshold $h_1=5\times 10^{-3}$, which yields no false alarm across $150$ honest streams of $5000$ steps (null ARL $> 5000$). A covert drift of magnitude $a$ per step in a behaviorally-active direction is injected at change point $t_0=50$. The CUSUM remains at zero before $t_0$ and fires shortly after: detection delays were $7$, $1$, and $0$ steps for $a\in\{0.02,0.05,0.10\}$ respectively, decreasing in drift magnitude as Lorden theory predicts. Crucially, the detector is fed $\delta_t$ (a per-step KL), not the Lyapunov level $V$; feeding the absolute level $V$ --- which starts near $V_0\approx 2$ --- would trip any threshold below $V_0$ at $t=1$ regardless of drift and is not a detection. We also note that a constant offset applied uniformly across the $K$ output logits lies in the softmax gauge null space and produces \emph{zero} behavioral change; an adversarial drift must therefore be injected in a non-gauge direction to be both effective and detectable. Panels A5a and A5b of Figure~\ref{fig:numerical_a} display the Lyapunov decay and the CUSUM path with its change point and trigger.

\begin{figure}[ht]
\centering
\includegraphics[width=\textwidth]{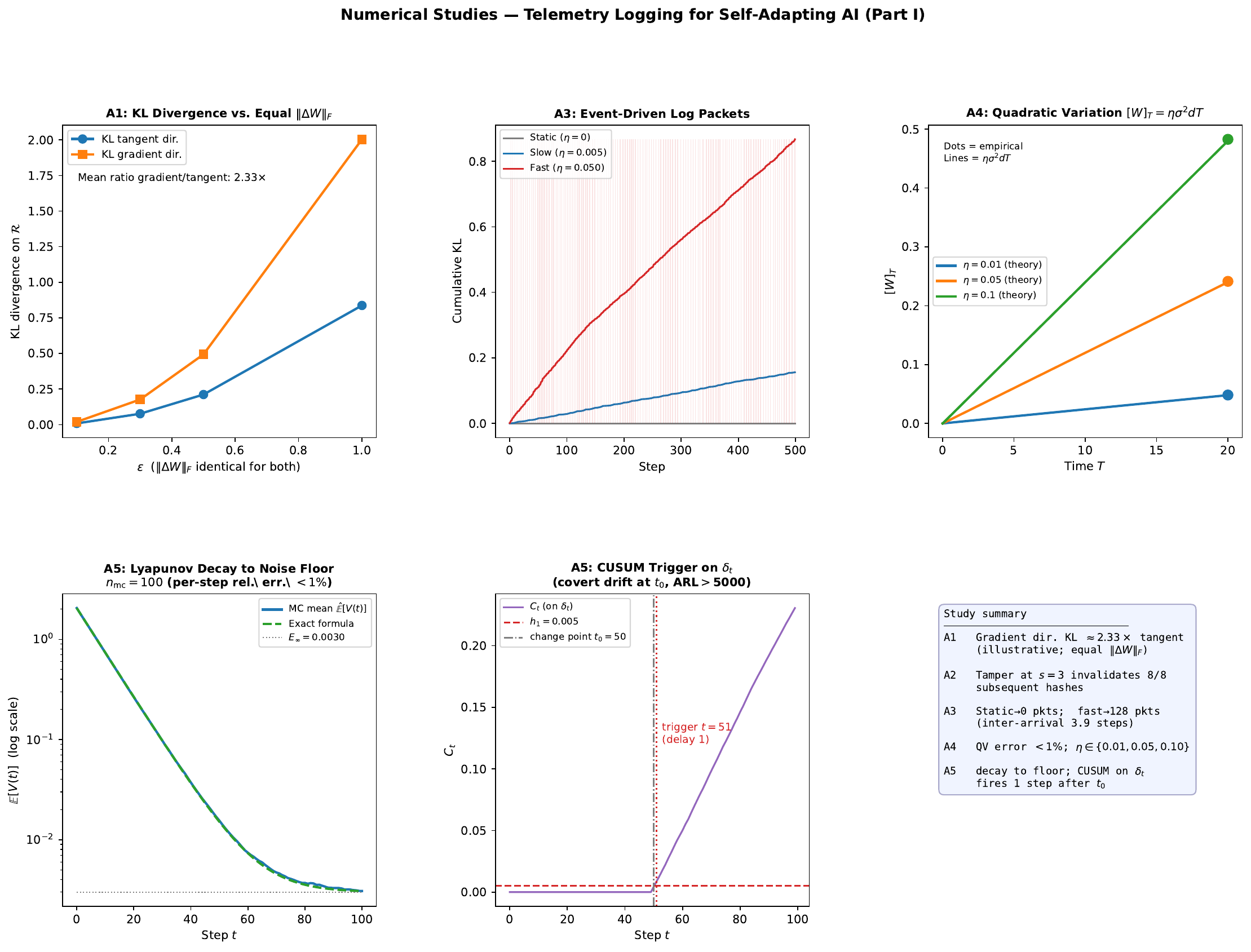}
\caption{Numerical studies for the Part~I telemetry logging framework. \textbf{A1:} KL divergence vs.\ equal-$\|\Delta W\|_F$ perturbations in the gradient and tangent directions; gradient direction produces $2.33\times$ higher KL at every $\varepsilon$. \textbf{A3:} Cumulative KL and event-driven log-packet arrivals (red ticks) for static, slow, and fast learning rates; packet density scales with adaptation intensity. \textbf{A4:} Empirical quadratic variation vs.\ theoretical $\eta\sigma^2 d T$ for three learning rates; errors $<1\%$. \textbf{A5 (left):} MC-averaged Lyapunov envelope $\hat{\mathbb{E}}[V(t)]$ vs.\ exact formula on a log scale; relative error $\leq1.04\%$ at the tabulated times. \textbf{A5 (right):} CUSUM statistic under adversarial drift; trigger fires at $t=1$.}
\label{fig:numerical_a}
\end{figure}


\section{The Model Hiding Problem}
\label{sec:hiding}

Part~I established a rigorous telemetry architecture for self-adapting generative AI models: how to log weight trajectory changes in discrete and continuous time in a manner that is both informationally sufficient for audit and computationally tractable. That framework assumes the firm transmits the telemetry faithfully. Part~II relaxes that assumption.

A firm deploying a large language model (LLM) or other generative AI system in a governed context --- credit decisions, insurance underwriting, algorithmic trading, clinical decision support --- faces a fundamental conflict of interest. Continual learning improves model performance and competitive advantage. Reporting that learning to validators triggers validation requirements, potentially halting the model until a review is complete. The firm therefore has an incentive to learn without reporting.

We term this the \emph{Model Hiding Problem}: the adversarial game between a model owner that updates its model covertly and a validation function attempting to detect such updates from observable signals alone.

\subsection{Related Work}

The Model Hiding Problem is related to but distinct from several bodies of literature. The \emph{model watermarking} literature \citep{adi2018turning, zhang2018protecting} focuses on proving model ownership, not detecting covert updates. The \emph{machine unlearning} literature \citep{cao2015towards} addresses removing information from trained models. The \emph{auditing} literature in algorithmic fairness \citep{raji2020closing} addresses outcome-based testing but does not address adversarial evasion of such tests. The \emph{Byzantine fault tolerance} literature \citep{lamport1982byzantine} addresses distributed systems with malicious nodes, which provides useful cryptographic primitives but does not model the specific incentive structure of governance evasion.

\section{The Concealment Game}
\label{sec:game}

\subsection{Formal Setup}

We model the interaction between the independent validation function and the model owner as an adversarial game on the weight process $\{W_t\}_{t \geq 0}$.

\begin{definition}[Concealment Game]
\label{def:game}
The game has two players:
\begin{itemize}
    \item \textbf{The Validator} $\mathcal{V}$ (the independent validation function): observes the telemetry stream $\mathcal{T} = \{\mathcal{D}_{\tau_k}\}_{k \geq 0}$ and the production output stream $\{y_t\}$ from periodic canary queries. The validator's goal is to detect any update $W_t \neq W_{\tau_{k^*}}$ where $\tau_{k^*}$ is the time of the most recent reported model state, with minimal delay.
    \item \textbf{The Model Owner} $\mathcal{F}$: controls the weight trajectory $\{W_t\}$, the telemetry stream $\mathcal{T}$, and the training data $\{\mathcal{B}_t\}$. The firm's goal is to maximize the cumulative undetected drift $\int_0^T V(W_t)\, dt$ subject to not being detected.
\end{itemize}
\end{definition}

\begin{definition}[Attack Surface]
An \emph{attack} is a strategy for the model owner that increases undetected drift while passing all validation checks. We classify attacks by their target:
\begin{itemize}
    \item \textbf{Attack S (Suppression):} Suppress the telemetry stream entirely.
    \item \textbf{Attack F (Forgery):} Transmit forged telemetry that shows compliant behavior while the model drifts.
    \item \textbf{Attack C (Channel):} Update via a side-channel not monitored by the probe.
    \item \textbf{Attack N (Noise):} Exploit the noise structure of the learning process to mask drift in signal.
    \item \textbf{Attack G (Girsanov):} Manipulate the training data distribution to induce a change of measure that makes the trajectory appear compliant.
    \item \textbf{Attack H (Hessian):} Engineer local curvature to exploit the second-order term of the It\^{o} expansion.
\end{itemize}
\end{definition}

Types S, F, C arise in discrete time. Types N, G, H are specific to continuous time. We address each in turn.

\section{Discrete-Time Attacks and Countermeasures}
\label{sec:discrete_attacks}

\subsection{Attack S: Suppression}

\subsubsection{Attack Description}

The firm updates $W_t$ continuously but transmits no log packets, declaring the model static: $W_t = W_0$ for all $t$. Under the basic telemetry architecture of Part~I, this is undetectable from the telemetry stream alone (since no stream exists to contradict the claim).

\subsubsection{Countermeasure: Behavioral Oracle with KL Detection}

The validator maintains a private canary set $\mathcal{C} = \{(x_i, p_{W_0}(\cdot \mid x_i))\}_{i=1}^m$ whose output distributions under the validated model $W_0$ are known. Crucially, $\mathcal{C}$ must be \emph{secret from the firm}; its governance value depends entirely on the firm's inability to overfit to it.

\begin{definition}[KL Oracle Statistic]
\begin{equation}
    T_t^{\mathrm{oracle}} = \sum_{i=1}^{|\mathcal{C}|} D_{\mathrm{KL}}\!\left(p_{W_0}(\cdot \mid x_i)\;\big\|\; p_{W_t}(\cdot \mid x_i)\right)
    \label{eq:kl_oracle}
\end{equation}
\end{definition}

The KL oracle statistic satisfies $T_t^{\mathrm{oracle}} \geq 0$ with equality if and only if $p_{W_t}(\cdot\mid x_i) = p_{W_0}(\cdot\mid x_i)$ for all $x_i\in\mathcal{C}$ (Gibbs' inequality). We adopt it in place of the unbiased MMD$^2$ estimator \citep{gretton2012kernel} for a specific operational reason. The unbiased MMD$^2$ has expectation \emph{exactly zero} under $H_0$ (that is what unbiasedness means), but because it omits the diagonal kernel terms it frequently takes \emph{negative values} on individual finite samples; a null threshold calibrated as an upper quantile of that distribution can therefore be negative, so a test of the form $\widehat{\mathrm{MMD}}^2 > c_\alpha$ admits configurations in which a non-drifted model is flagged or, conversely, the threshold loses interpretability as a divergence. The KL statistic is non-negative by construction, yielding a one-sided test with a non-negative threshold and an unambiguous zero at $H_0$. (This is a matter of estimator convenience, not a claim that MMD fails as a two-sample test.)

\begin{theorem}[Oracle Detection Consistency]
\label{thm:oracle}
Under $H_0$ (the model is static, $W_t = W_0$ for all $t$), $T_t^{\mathrm{oracle}} = 0$ exactly.  Under $H_1$ (the model has drifted so that $p_{W_t}(\cdot\mid x)\neq p_{W_0}(\cdot\mid x)$ for some $x$), $T_t^{\mathrm{oracle}} > 0$ provided the canary set is behaviorally sufficient (Definition~\ref{def:reference_set}).  The test $\phi_t = \mathbf{1}[T_t^{\mathrm{oracle}} > c_\alpha]$ is consistent at level $\alpha$, where $c_\alpha$ is the $\alpha$-upper quantile of the null bootstrap distribution obtained by evaluating \eqref{eq:kl_oracle} under small hardware-level weight noise.
\end{theorem}

\begin{proof}
Under $H_0$, $W_t = W_0$ so $p_{W_t}(\cdot\mid x_i) = p_{W_0}(\cdot\mid x_i)$ for all $i$, giving $T_t^{\mathrm{oracle}} = 0$.  Under $H_1$, by the non-negativity and faithfulness of the KL divergence, $D_{\mathrm{KL}}(p_{W_0}(\cdot\mid x_i)\|p_{W_t}(\cdot\mid x_i)) > 0$ for any $x_i$ at which the output distribution has changed.  Since the canary set is behaviorally sufficient, at least one such $x_i$ exists, giving $T_t^{\mathrm{oracle}} > 0$.  Bootstrap calibration under small noise ensures $c_\alpha \geq 0$, so the test rejects at rate converging to one as $\|W_t - W_0\|\to\infty$, establishing consistency.
\end{proof}

Theorem~\ref{thm:oracle} is what makes concealment \emph{contestable}: a model that has genuinely drifted cannot present as static on a behaviorally sufficient canary set, so undisclosed adaptation leaves a detectable trace that the validation function can act on. The secrecy requirement below is the price of that guarantee, and its fragility previews the central limit of this work: detection is possible, but only down to a resolution floor.

\begin{remark}[The Secrecy Requirement]
If the canary set $\mathcal{C}$ is known to the firm, they can train the model to maintain correct outputs on $\mathcal{C}$ while drifting on the unmonitored remainder of the input space. The secrecy of $\mathcal{C}$ is not a desideratum but a hard requirement: the detection guarantee of Theorem~\ref{thm:oracle} fails under any partial disclosure. Operationally, this means $\mathcal{C}$ must be held in a governance-controlled enclave and rotated periodically to prevent adaptive evasion.
\end{remark}

\subsection{Attack F: Lyapunov Forgery}

\subsubsection{Attack Description}

The firm runs the real update $W_t \to \hat{W}_{t+1}$ (which violates the Lyapunov condition $\Delta V \leq -\alpha V$), simultaneously computes a compliant ``decoy'' state $W_{t+1}^{\mathrm{decoy}}$ that satisfies the Lyapunov check, and reports $h_{t+1} = \mathrm{Hash}(W_{t+1}^{\mathrm{decoy}})$ to the ledger while serving inference from $\hat{W}_{t+1}$.

\subsubsection{Countermeasure 1: Inference Signing}

Every production inference output is signed with a commitment to the current model hash:

\begin{equation}
    \mathrm{Sig}_t(y) = \mathrm{Sign}_{\mathrm{sk}}\!\left(y \;\big\|\; h_t \;\big\|\; t\right)
\end{equation}

where $\mathrm{sk}$ is a secret key held by the TEE (see Attack~C countermeasure). This binds every production output to the declared model version. The validator can spot-check: given any signed output, do the canary responses under $W_{t+1}^{\mathrm{decoy}}$ match the signed output? If the signed hash corresponds to $W_{t+1}^{\mathrm{decoy}}$ but the canary output is statistically consistent with $\hat{W}_{t+1}$, there is a provable contradiction between the declared model and the serving model. This binds what a firm \emph{serves} to what it \emph{reports}---a transparency property that denies a model owner any place to conceal divergence between disclosed and actual behavior.

\begin{proposition}[Forgery Detection via Consistency Test]
\label{prop:forgery}
Suppose the firm reports $h_{t+1} = \mathrm{Hash}(W_{t+1}^{\mathrm{decoy}})$ but serves inference from $\hat{W}_{t+1}$ where $\hat{W}_{t+1} \neq W_{t+1}^{\mathrm{decoy}}$. Then:
\begin{equation}
    \Pr\!\left[T_t^{\mathrm{oracle}}\big(W_{t+1}^{\mathrm{decoy}}\big) \leq c_\alpha\right] \leq \alpha + \mathrm{negl}(\lambda)
\end{equation}
under the behavioral oracle test, provided $\|\hat{W}_{t+1} - W_{t+1}^{\mathrm{decoy}}\|$ exceeds the detection threshold of Theorem~\ref{thm:oracle}.
\end{proposition}

\begin{proof}
The declared object is $W_{t+1}^{\mathrm{decoy}}$; the serving object is $\hat{W}_{t+1}$. The
validator evaluates the behavioral oracle statistic $T_t^{\mathrm{oracle}}$ on the canary set using
outputs produced by the \emph{serving} model, but compares them against the distribution induced by
the \emph{declared} model. Two cases exhaust the possibilities.

\emph{Case 1: the firm serves what it declares.} Then $\hat{W}_{t+1} = W_{t+1}^{\mathrm{decoy}}$,
contradicting the hypothesis; nothing to prove.

\emph{Case 2: the firm serves $\hat W_{t+1}\neq W_{t+1}^{\mathrm{decoy}}$.} By hypothesis the
separation $\|\hat{W}_{t+1} - W_{t+1}^{\mathrm{decoy}}\|$ exceeds the detection threshold of
Theorem~\ref{thm:oracle}, so on a behaviorally sufficient canary set the induced output
distributions differ in KL divergence by a strictly positive margin. By the consistency statement of
Theorem~\ref{thm:oracle}, the test therefore rejects with probability at least $1-\alpha$ under the
calibration $c_\alpha$; equivalently the acceptance probability is at most $\alpha$.

The additive $\mathrm{negl}(\lambda)$ accounts for the only remaining escape route: producing a
different pre-image with the same declared hash. Since $\mathrm{Hash}$ is collision-resistant with
security parameter $\lambda$, a probabilistic polynomial-time firm succeeds with probability at most
$\mathrm{negl}(\lambda)$. Summing the two contributions gives the stated bound.

The hypothesis is essential and is not cosmetic: if the separation is \emph{below} the oracle's
detection threshold, the proposition asserts nothing --- which is exactly the residual quantified in
Section~\ref{sec:residual}.
\end{proof}

\subsubsection{Countermeasure 2: Zero-Knowledge Proof of Lyapunov Compliance}

A more powerful countermeasure converts the Lyapunov check from self-reported compliance into a cryptographically verifiable statement. The firm generates a zero-knowledge succinct non-interactive argument of knowledge (zk-SNARK) \citep{groth2016size} at each reporting step:

\begin{definition}[ZK Lyapunov Proof]
The firm generates a proof $\pi_t$ attesting to the following statement without revealing $W_t$:
\begin{quote}
\textit{``There exists $W_{t+1}$ with $\mathrm{Hash}(W_{t+1}) = h_{t+1}$ such that $V(W_{t+1}) \leq V(W_t) - \alpha \|W_t - \proj_{\mathcal{S}}(W_t)\|^2$.''}
\end{quote}
The verifier (the validation function) checks $\pi_t$ efficiently without weight access. Under the soundness of the underlying proof system, generating a valid $\pi_t$ for a false statement is computationally infeasible.
\end{definition}

\begin{remark}[Practical Considerations]
Generating zk-SNARKs for large neural network weight matrices is computationally expensive with current technology. For Tier~3 (Macro) updates, full ZK compliance proofs are feasible and warranted. For Tier~1 (Micro) updates, inference signing combined with periodic oracle testing is a computationally tractable substitute.
\end{remark}

\subsection{Attack C: Shadow Model Side-Channel}

\subsubsection{Attack Description}

The firm maintains two models: a reported model $M_r$ whose telemetry is faithfully submitted, and an actual serving model $M_a$ that is updated without reporting. Production outputs are a blend:
\begin{equation}
    y_t = \lambda\, f(x_t;\, M_a) + (1-\lambda)\, f(x_t;\, M_r)
\end{equation}
for some blending parameter $\lambda \in (0,1)$. The telemetry reflects only $M_r$, which passes all checks.

\subsubsection{Countermeasure: Trusted Execution Environment (TEE) Attestation}

\begin{definition}[TEE Learning Attestation]
The entire training pipeline --- gradient computation, weight update, and inference serving --- is executed within a Trusted Execution Environment (TEE) such as Intel TDX or AMD SEV-SNP. At each update step, the TEE generates a signed attestation report:
\begin{equation}
    \mathcal{A}_t = \mathrm{TEE\_Sign}\!\left(\mathrm{Hash}(W_{t-1}),\; \mathrm{Hash}(\Delta W_t),\; \mathrm{Hash}(\mathcal{B}_t)\right)
\end{equation}
The attestation chain $\{\mathcal{A}_t\}_{t \geq 0}$ is transmitted to the governance ledger. The TEE's signing key is hardware-bound; the firm cannot modify weights outside the enclave without breaking the chain.
\end{definition}

Running an unreported model alongside the disclosed one is the analogue of keeping a second, hidden set of books: the model owner presents compliant telemetry while serving from something else. Hardware-rooted attestation closes this channel by binding the record to an enclave the owner cannot bypass, as the following completeness statement makes precise.

Before stating the guarantee we fix the threat model, because the claim is exhaustive only relative
to it.

\begin{assumption}[TEE threat model]
\label{ass:tee}
We assume: \textbf{(T1)} the enclave is correctly implemented and not compromised --- its signing key
is not extractable and its measurement of the loaded weight state is faithful; \textbf{(T2)} every
production output presented to the validator carries a TEE signature, so unsigned outputs are
rejected by construction rather than merely noticed; and \textbf{(T3)} the served output is the
enclave's output, i.e.\ no post-processing is applied to attested outputs after they leave the
enclave and before they reach the consumer.
\end{assumption}

\begin{proposition}[TEE Completeness under Assumption~\ref{ass:tee}]
Under Assumption~\ref{ass:tee}, any production output $y_t$ generated outside the attested weight state $W_t$ is detectable:
\begin{enumerate}[label=(\roman*)]
    \item If $M_a$ runs outside the TEE, its outputs cannot be signed by the TEE key, and inference signing detects the discrepancy.
    \item If $M_a$ runs inside the TEE, the TEE attestation $\mathcal{A}_t$ must include $\mathrm{Hash}(\Delta W_t^a)$, which is submitted to the ledger and thus disclosed to the validator.
\end{enumerate}
There is no configuration in which $M_a$ can serve blended outputs without either breaking the TEE attestation chain or disclosing the update to the ledger.
\end{proposition}

\begin{proof}
By (T2) every accepted output is signed, and by (T1) only the enclave can produce a valid signature,
so any served output was produced inside the enclave on a weight state the enclave measured. The
enclave's measurement is included in the attestation $\mathcal{A}_t$ and submitted to the ledger.
Thus for a given output either (i) it is unsigned or carries an invalid signature --- rejected under
(T2) --- or (ii) it is validly signed, in which case the weight state that produced it, including any
increment $\Delta W_t^a$, appears in the attested chain and is disclosed. The two cases are
exhaustive by construction, which gives the claim.
\end{proof}

\begin{remark}[What the assumption excludes, and why it matters]
\label{rmk:tee_scope}
The exhaustiveness above is relative to Assumption~\ref{ass:tee}, and each clause excludes a real
attack. Relaxing \textbf{(T1)} --- a compromised or side-channel-vulnerable enclave --- voids the
guarantee entirely, since a firm able to forge signatures can attest to any state it likes; hardware
trust is an assumption of this framework, not a result of it. Relaxing \textbf{(T3)} is subtler and
is the gap most likely to arise in practice: a firm may run an attested model inside the enclave and
then transform its outputs \emph{outside} it, so every signature verifies while the served behavior
differs from the attested behavior. Detecting that requires the behavioral oracle of
Section~\ref{sec:discrete_attacks} applied to \emph{served} outputs rather than enclave outputs ---
which is exactly the consistency test of Proposition~\ref{prop:forgery}. TEE attestation and
behavioral testing are therefore complements, not substitutes: neither closes the other's gap.
\end{remark}

\begin{remark}[Hardware Trust Assumption]
The TEE countermeasure assumes that the hardware manufacturer (Intel, AMD) has not been compromised and that the TEE implementation is correct. Both assumptions are strong. In practice, TEE attestation should be combined with behavioral oracle testing as a defense-in-depth strategy.
\end{remark}

\section{Continuous-Time Attacks and Countermeasures}
\label{sec:continuous_attacks}

The continuous-time limit introduces three attacks with no discrete-time analog, arising from the richer structure of It\^{o} diffusions.

\subsection{Attack N: Diffusion-Amplified Drift Masking}

\subsubsection{Attack Description}

Recall the weight SDE of Part~I:
\begin{equation}
    dW_t = \mu(W_t)\, dt + \sigma(W_t)\, dB_t
\end{equation}
where $\mu = -\nabla J$ is the drift and $\sigma = \sqrt{\eta}\, \Sigma^{1/2}$ is the diffusion coefficient.

The firm decomposes the drift into a legitimate component and a hidden adversarial component:
\begin{equation}
    dW_t = \left(\mu_{\mathrm{safe}}(W_t) + \mu_{\mathrm{hidden}}(W_t)\right) dt + \sigma(W_t)\, dB_t
\end{equation}

By temporarily amplifying $\sigma$ (e.g., inflating $\eta$ or increasing batch noise), the firm reduces the signal-to-noise ratio of $\mu_{\mathrm{hidden}}$. The time required to detect a hidden drift of magnitude $|\mu_{\mathrm{hidden}}|$ from an observed path scales as:
\begin{equation}
    T_{\mathrm{detect}} \sim \frac{\sigma^2}{|\mu_{\mathrm{hidden}}|^2}
    \label{eq:detect_time}
\end{equation}
Doubling $\sigma$ quadruples the detection time. The firm can transiently amplify noise while running $\mu_{\mathrm{hidden}}$, then reduce $\sigma$ once the drift has accumulated, leaving a weight trajectory that appears consistent with normal stochastic learning.

\subsubsection{Attempted Countermeasure: Power Variation --- and Why It Fails}

A natural idea is to separate the adversarial drift from the (possibly amplified) diffusion using \emph{power variations}: the first-order variation should grow with the drift while the quadratic variation tracks only the diffusion. We show that this does not work, and that the failure is intrinsic rather than a matter of tuning.

\begin{proposition}[Power Variation Does Not Detect Drift]
\label{thm:power_var}
Let $\{W_t\}$ solve the weight SDE on $[0,T]$, sampled on a uniform grid with spacing $\Delta=T/n$. Let
\[
V_T^{(1)} = \sum_{i=1}^{n}\|W_{t_i}-W_{t_{i-1}}\|_1
\quad\text{(realized first variation)},
\qquad
V_T^{(2)} = \sum_{i=1}^{n}\|W_{t_i}-W_{t_{i-1}}\|_2^2 \xrightarrow[\Delta\to0]{} [W]_T .
\]
Then, for \emph{any} drift $\mu$ (including $\mu\equiv 0$),
\begin{equation}
    \frac{V_T^{(1)}}{\sqrt{V_T^{(2)}}} = \Theta\!\big(\sqrt{n}\big) = \Theta\!\Big(\sqrt{T/\Delta}\Big),
    \label{eq:power_ratio}
\end{equation}
with the drift entering only through an $O(\sqrt{\Delta})$ relative correction that vanishes as $\Delta\to0$. The ratio therefore measures the \emph{sampling frequency}, not the drift.
\end{proposition}

\begin{proof}
The continuous first variation $\int_0^T \|dW_t/dt\|\,dt$ is \emph{infinite}: Brownian paths have unbounded variation on every interval --- the defining fact that forces It\^o, rather than Riemann--Stieltjes, calculus. Discretely, each increment is $W_{t_i}-W_{t_{i-1}} = \mu_{t_{i-1}}\Delta + \sqrt{\eta}\,\Sigma^{1/2}\Delta B_i$ with $\Delta B_i\sim\mathcal N(0,\Delta I)$, so the diffusive part dominates the increment ($\sqrt{\Delta}\gg\Delta$). Hence
\[
\mathbb{E}\,V_T^{(1)} = \sqrt{\tfrac{2}{\pi}}\Big(\textstyle\sum_j\sqrt{\eta\Sigma_{jj}}\Big)\sqrt{\Delta}\cdot n + O(\|\mu\|_1 T) = \Theta\!\big(\sqrt{\Delta}\cdot n\big) = \Theta\!\big(\sqrt{nT}\big),
\]
while $\sqrt{V_T^{(2)}}=\sqrt{[W]_T}+o(1)=\Theta(\sqrt{T})$. Their ratio is $\Theta(\sqrt{n})$. The drift contributes only the additive $O(\|\mu\|_1 T)$ term to the numerator, a relative correction of order $T/\sqrt{nT}=\sqrt{T/n}=\sqrt{\Delta}\to0$.
\end{proof}

\begin{remark}[The correct object, and the deeper obstruction]
\label{rmk:total_drift}
The drift signal lives not in the absolute variation but in the \emph{signed} (telescoping) sum $W_T-W_0 = \int_0^T\mu_s\,ds + \int_0^T\sqrt{\eta}\,\Sigma^{1/2}dB_s$. The standardized net displacement $(W_T-W_0)/\sqrt{[W]_T}$ is asymptotically $\mathcal N(0,1)$ when the average drift is zero and grows like $(\bar\mu/\sigma)\sqrt{T}$ otherwise, so it separates drift \emph{only as $T\to\infty$}, requiring a horizon $T_{\mathrm{detect}}\sim\sigma^2/|\mu|^2$ (Eq.~\eqref{eq:detect_time}). Two consequences are fatal for a standalone Attack-N countermeasure. First, over a \emph{fixed} validation horizon a bounded total drift produces an arbitrarily small standardized signal. Second, and more fundamentally, the variance of any drift estimator over a fixed window is $\Theta(\sigma^2/T)$ \emph{independent of the sampling rate}: refining $\Delta$ cannot improve drift estimation (this is the classical hardness of estimating drift from high-frequency data --- power variations separate diffusion from \emph{jumps}, not drift). We therefore \textbf{do not claim a working path-statistic countermeasure for Attack~N}. The inability to certify the absence of small persistent drift from telemetry alone is not a gap to be patched with a better statistic; it is a concrete instance of the information-theoretic residual of Section~\ref{sec:residual}, and it is properly handled there and by periodic invasive audits.
\end{remark}

\subsection{Attack G: Girsanov Measure Manipulation}

\subsubsection{Attack Description}

This is the most technically sophisticated attack. By Girsanov's theorem, a change of drift in an SDE is equivalent to a change of probability measure.

The significance for governance is that concealment need not be crude. A firm can make it mathematically subtle---a change of the measure under which its own dynamics appear benign---so the validation function's task is not merely to catch missing reports but to contend with manipulation at the level of the probability law itself. This is part of why we ultimately concede an irreducible residual rather than claiming complete detectability.

\begin{theorem}[Girsanov's Theorem \citep{girsanov1960transforming}; governance formulation]
\label{thm:girsanov}
\emph{(Classical; stated without proof. See \cite{oksendal2003stochastic} Ch.~8 for a textbook treatment; only the interpretation below is ours.)}
Suppose the firm chooses training data $\{\mathcal{B}_t\}$ such that the gradient process induces a change of measure from $\mathbb{P}$ (the natural measure under which the validator audits) to $\mathbb{Q}$. The Radon--Nikodym derivative is:
\begin{equation}
    \frac{d\mathbb{Q}}{d\mathbb{P}}\bigg|_{\mathcal{F}_T} = \exp\!\left(\int_0^T h_s\, dB_s - \frac{1}{2}\int_0^T |h_s|^2\, ds\right)
    \label{eq:girsanov}
\end{equation}
where $h_s$ encodes the adversarial drift induced by the data manipulation. Under $\mathbb{Q}$, the Lyapunov condition is satisfied and telemetry appears compliant. Under $\mathbb{P}$, the model drifts adversarially.
\end{theorem}

Concretely: the firm selects training batches $\mathcal{B}_t$ whose gradient $\nabla_W J(W_t, \mathcal{B}_t)$ includes a component $h_t$ that, when accumulated over time, moves the model toward a target drifted state. But the gradient covariance $\Sigma_t$ is engineered to be large, masking the drift component in noise. From the validator's perspective (observing outputs and telemetry under $\mathbb{P}$), the trajectory appears consistent with normal learning.

\subsubsection{Attempted Countermeasure: Gradient Injection --- and Why It Is Circular}

The natural countermeasure is for the validator to inject training examples with a known expected gradient and check the model's response. We define it, then explain why it does not, as stated, detect the Girsanov attack.

\begin{definition}[Gradient Probe]
At random times $\{s_j\}$, the validator injects a \emph{probe batch} $\tilde{\mathcal{B}}_j$ with known expected gradient $g^* = \mathbb{E}[\nabla_W J(W_0, \tilde{\mathcal{B}}_j)]$ under the validated model $W_0$. The actual gradient response $\hat{g}_j = \nabla_W J(W_{s_j}, \tilde{\mathcal{B}}_j)$ is computed at the current weights and returned to the validator.
\end{definition}

\begin{proposition}[The Gradient Probe Does Not Test the Measure Change]
\label{thm:gradient_probe}
The probe response $\hat g_j = \nabla_W J(W_{s_j},\tilde{\mathcal B}_j)$ is a \emph{deterministic} function of the current weights $W_{s_j}$ and the (validator-chosen) probe batch $\tilde{\mathcal B}_j$. It does not depend on the firm's training measure $\mathbb{Q}$ at all, except through the value of $W_{s_j}$ that the firm's training has already produced. Consequently:
\begin{enumerate}[label=(\roman*)]
    \item Under a firm that manipulates only its \emph{training} distribution, $\mathbb{E}[\hat g_j\mid W_{s_j}] = \nabla_W J(W_{s_j},\tilde{\mathcal B}_j)$, which equals $g^*$ whenever $W_{s_j}=W_0$ and otherwise differs by exactly the drift already incurred --- not by the Girsanov integrand $h_{s_j}$. The identity ``$\mathbb{E}[\hat g_j\mid W_{s_j}]=g^*+h_{s_j}$'' holds only if the firm additionally perturbs gradient \emph{evaluations on the validator's own probe}, which a training-data manipulator does not do.
    \item Whatever the probe does detect is a function of $W_{s_j}\neq W_0$, i.e.\ that the weights have already moved --- precisely what the behavioral oracle (Theorem~\ref{thm:oracle}) already detects from outputs.
\end{enumerate}
\end{proposition}

\begin{proof}
$\nabla_W J(\cdot,\tilde{\mathcal B}_j)$ is a fixed function; evaluated at $W_{s_j}$ it returns a fixed vector. The training measure enters the joint law only through the distribution of $W_{s_j}$. Hence conditional on $W_{s_j}$ the probe carries no information about $\mathbb{Q}$ beyond $W_{s_j}$ itself, and any test built on $\|\hat g_j-g^*\|^2$ is a (noisier) function of the same weight displacement the output oracle already sees.
\end{proof}

\begin{remark}[Status of Attack~G]
We therefore \textbf{demote} the gradient probe: it is not an independent countermeasure to a Girsanov-type measure manipulation, and at best re-detects weight drift already caught by the behavioral oracle. (Eq.~\eqref{eq:girsanov} requires Novikov's condition $\mathbb{E}[\exp\tfrac12\int_0^T|h_s|^2 ds]<\infty$ for the change of measure to be a genuine probability measure; even granting it, the attack's effect is a change in the realized weight \emph{trajectory}, which the oracle monitors directly.) As with Attack~N, the residual risk --- a measure manipulation whose induced output change stays below the behavioral resolution --- is governed by the information-theoretic residual of Section~\ref{sec:residual}, not by gradient injection. A genuinely independent gradient-space test would have to interrogate the firm's \emph{training-time} gradient stream under attestation (Attack~C, TEE), not a post-hoc probe of the current weights.
\end{remark}

\subsection{Attack H: Hessian Exploitation via Noise Energy}

\subsubsection{Attack Description}

Returning to the It\^{o}--Lyapunov expansion of Part~I:
\begin{equation}
    \mathbb{E}[dV(W_t)] = \underbrace{-\langle \nabla V,\, \nabla J \rangle\, dt}_{\text{deterministic drift}} + \underbrace{\frac{\eta}{2}\tr(\Sigma\, \nabla^2 V)\, dt}_{\text{noise energy term}}
    \label{eq:ito_lyap2}
\end{equation}

The noise energy term $\frac{\eta}{2}\tr(\Sigma\, \nabla^2 V)$ depends on the Hessian of the Lyapunov function $\nabla^2 V$. When the safe manifold $\mathcal{S}$ has high curvature, $\nabla^2 V$ has large eigenvalues, and the noise energy term can be made either large and positive or large and negative by controlling $\Sigma$ (equivalently, controlling the gradient covariance via batch selection).

The firm engineers a situation where the deterministic drift is positive ($-\langle \nabla V, \nabla J \rangle > 0$: the model is drifting away from $\mathcal{S}$) but offset by a temporarily negative noise energy term (achieved by reducing $\sigma$ near the boundary of $\mathcal{S}$, which flattens the effective curvature of $V$ in the active direction). The \emph{mean stability condition}
\begin{equation}
    -\langle \nabla V, \nabla J \rangle + \frac{\eta}{2}\tr(\Sigma\, \nabla^2 V) \leq -\alpha V(W_t)
    \label{eq:mean_stable_b}
\end{equation}
is satisfied in expectation, while individual trajectories escape $\mathcal{S}$ through a low-noise corridor.

\subsubsection{Countermeasure: Almost-Sure Stability via the Law of the Iterated Logarithm}

Mean stability is insufficient to close this attack. The required condition is almost-sure path-wise stability.

\begin{definition}[Path-Wise Stability Condition]
Let $M_t = \int_0^t \nabla V^\top \sqrt{\eta}\, \Sigma^{1/2}\, dB_s$ be the martingale term of \eqref{eq:ito_lyap2}, with $[M]_t = \int_0^t \eta\, \|\Sigma^{1/2} \nabla V\|^2\, ds$. The \emph{path-wise Lyapunov stability certificate} consists of two monitorable conditions:
\begin{align}
    \text{(drift certificate)}\quad & -\langle \nabla V, \nabla J \rangle + \tfrac{\eta}{2}\tr(\Sigma\, \nabla^2 V) \;\leq\; -\alpha\big(V(W_t)-V_\infty\big), \label{eq:as_lyap}\\
    \text{(envelope check)}\quad & \Big|\,V(W_t) - V_0 - {\textstyle\int_0^t}\!\big[-\langle\nabla V,\nabla J\rangle+\tfrac{\eta}{2}\tr(\Sigma\nabla^2 V)\big]ds\,\Big| \nonumber\\
    &\qquad\qquad \;\le\; (1+\varepsilon)\sqrt{2[M]_t\log\log [M]_t}\ \ \text{a.s.\ eventually,} \label{eq:as_envelope}
\end{align}
where both $\Sigma_t$ and $[M]_t$ are evaluated from the firm's \emph{declared} telemetry.
\end{definition}

\begin{remark}
Note the structure: the LIL term appears as a \emph{two-sided envelope} on the realized residual \eqref{eq:as_envelope}, not added to the drift in a one-sided inequality. The naive form ``$-\langle\nabla V,\nabla J\rangle+\tfrac{\eta}{2}\tr(\Sigma\nabla^2 V)+C\sqrt{[M]_t\log\log[M]_t}\le-\alpha V$'' is \emph{unsatisfiable for large $t$}: $[M]_t\to\infty$ generically, so the positive LIL term grows without bound and cannot be dominated by a negative right-hand side. The correct content of the LIL here is that the martingale fluctuations stay \emph{within} the envelope of the declared $[M]_t$; a firm whose realized path leaves that envelope has mis-declared its second-order telemetry.
\end{remark}

Attack~H is the most insidious of the concealment strategies because it disguises a move toward instability as ordinary learning noise---an adverse drift run through a low-noise corridor so that mean-stability checks still pass. Closing it ensures a firm cannot present destabilizing adaptation as benign stochastic fluctuation, preserving the integrity of the stability certificates on which the governance framework rests.

\begin{theorem}[LIL/Envelope Closure of Attack H]
\label{thm:lil}
Suppose the firm satisfies the declared drift certificate \eqref{eq:as_lyap}. Then $V(W_t)$ is ultimately bounded, $\limsup_{t\to\infty}V(W_t)\le V_\infty$ a.s.\ (the ultimate-boundedness theorem of Part~I). An Attack-H firm that runs an adverse drift through a low-noise corridor faces a dichotomy:
\begin{enumerate}[label=(\roman*)]
    \item \textbf{Truthful $[M]_t$:} the adverse drift makes the realized residual leave the envelope \eqref{eq:as_envelope} (the LIL bound is sharp, $\limsup |M_t|/\sqrt{2[M]_t\log\log[M]_t}=1$ a.s.), so the violation is detected.
    \item \textbf{Under-declared $[M]_t$:} suppressing the declared noise to keep the corridor narrow shrinks the declared envelope below the realized fluctuations, so the residual again leaves the band (Study~B4, where an under-declared $\sigma$ produces a $99.7\%$ violation rate against a $0.9\%$ honest rate).
\end{enumerate}
The firm cannot simultaneously (a) run an adverse deterministic drift, (b) declare second-order telemetry consistent with its realized path, and (c) pass both \eqref{eq:as_lyap} and \eqref{eq:as_envelope}.
\end{theorem}

\begin{proof}
Decompose the realized path as in Corollary~\ref{cor:as_ultimate}:
$V(W_t) = V(W_0) + \int_0^t \mathcal{L}V(W_s)\,ds + M_t$, and write $\widehat{[M]}_t$ for the
\emph{declared} quadratic variation and $[M]_t$ for the realized one. The envelope test
\eqref{eq:as_envelope} is evaluated with $\widehat{[M]}_t$, since that is what the validator
observes; the path obeys the law of the iterated logarithm with the realized $[M]_t$.

Suppose the firm runs an adverse drift, i.e.\ there is $\epsilon_0>0$ and a set of times of positive
density on which $\mathcal{L}V(W_s) \geq -\alpha(V(W_s)-V_\infty) + \epsilon_0$. Then the drift
contribution to $V$ exceeds the certified one by an amount growing linearly in $t$ along those times.
Since the test compares the residual
$R_t := V(W_t) - V(W_0) - \int_0^t\big[-\alpha(V-V_\infty)\big]ds$ against the declared envelope, the
adverse drift adds a term of order $\epsilon_0 t$ to $R_t$. Two cases exhaust the firm's options.

\emph{(i) Truthful declaration, $\widehat{[M]}_t = [M]_t$.} By the law of the iterated logarithm for
continuous local martingales, $\limsup_{t\to\infty}|M_t|/\sqrt{2[M]_t\log\log[M]_t} = 1$ almost
surely, so the martingale part is eventually confined to the envelope
$(1+\varepsilon)\sqrt{2[M]_t\log\log[M]_t}$, which grows as $O(\sqrt{t\log\log t})$ under (H2). The
adverse drift contributes $\Omega(t)$, and $t/\sqrt{t\log\log t}\to\infty$; hence $R_t$ exits the
envelope for all large $t$ and the violation is detected with probability one.

\emph{(ii) Under-declared quadratic variation, $\widehat{[M]}_t < [M]_t$.} Shrinking the declared
noise narrows the corridor: the declared envelope is
$(1+\varepsilon)\sqrt{2\widehat{[M]}_t\log\log\widehat{[M]}_t}$, strictly smaller than the realized
one. But the realized fluctuations still satisfy the LIL with $[M]_t$, and by sharpness of the LIL
the realized $|M_t|$ attains $(1-\varepsilon')\sqrt{2[M]_t\log\log[M]_t}$ infinitely often. Since
$\widehat{[M]}_t<[M]_t$, for any $\varepsilon$ there are arbitrarily large $t$ at which
$|M_t|$ exceeds the declared envelope, so the residual leaves the band infinitely often and the
under-declaration is detected.

There is no third option: either the declared second-order telemetry matches the realized path or it
does not. Hence (a), (b) and (c) cannot hold simultaneously.

\emph{Scope.} The argument is asymptotic --- it establishes detection for all sufficiently large $t$
(case i) or infinitely often (case ii), not detection within a prescribed finite horizon. Detection
delay depends on $\epsilon_0$, on the gap $[M]_t-\widehat{[M]}_t$, and on the tolerance
$\varepsilon$; Study~B4 gives the finite-sample picture at the parameters tested (a $99.7\%$
violation rate against a $0.9\%$ honest rate).
\end{proof}

\begin{auditimp}[Path-wise certification]
\label{audit:pathwise}
This is a recommendation about audit design rather than a mathematical consequence, and is stated as
such. The validation audit condition should require the path-wise certificate \eqref{eq:as_lyap}--\eqref{eq:as_envelope}, not mean stability \eqref{eq:mean_stable_b} alone. The additional telemetry required is $[M]_t$, computable from the declared gradient covariance $\Sigma_t$ and $\nabla V(W_t)$; cross-checking the realized residual against the declared $[M]_t$ is what exposes both adverse drift and noise mis-declaration.
\end{auditimp}

\section{The Information-Theoretic Residual}
\label{sec:residual}

\subsection{Fundamental Limit}

Even with all countermeasures active, there is an irreducible residual that cannot be closed by telemetry alone. The negative results above are not isolated; they are symptoms of this single limit, which says that no disclosure regime, however well designed, can guarantee the absence of small, persistent adaptation. Surveillance has a resolution floor, and what lies beneath it can be reached only by stepping outside the telemetry channel---through invasive audit, or through structural limits on how the model is built.

\begin{heuristic}[Information-Theoretic Residual --- counting argument]
\label{thm:it_residual}
\emph{This is a motivating counting argument, not a theorem, and the paper's formal fundamental-limit
result is Theorem~\ref{cor:undetectable} below, which does not depend on it.} We state it because it
conveys the intuition compactly; the reasons it falls short of a theorem are given immediately after.
Let $\mathcal{O} = (\mathcal{T}, \{y_t^{\mathcal{C}}\})$ be the full observation of the validator, comprising the telemetry stream $\mathcal{T}$ and canary query outputs $\{y_t^{\mathcal{C}}\}$. Then:
\begin{equation}
    I\!\left(W_t \;;\; \mathcal{O}\right) < H(W_t)
    \label{eq:it_bound}
\end{equation}
That is, the mutual information between the weight trajectory and the observable stream would be strictly less than the entropy of the weight trajectory, leaving a nonzero residual $H(W_t) - I(W_t;\, \mathcal{O})$ that the validator cannot recover from telemetry and canary queries alone.
\end{heuristic}

\begin{proof}[Counting argument (not a proof)]
$H(W_t) = \Theta(d)$ where $d \sim 10^9$ to $10^{12}$ is the weight dimension. The telemetry stream $\mathcal{T}$ transmits $O(1)$ scalars per packet; the canary queries $\{y_t^{\mathcal{C}}\}$ provide $O(m)$ output samples per audit, carrying at most $O(m\log|\mathcal{Y}|)$ bits. Whenever the weight dimension exceeds the telemetry bandwidth, $d \gg m\log|\mathcal{Y}|$, the data-processing inequality gives $I(W_t;\mathcal{O}) \leq O(m\log|\mathcal{Y}|) \ll H(W_t)$, leaving a strictly positive residual.
\end{proof}

\begin{remark}[Why the counting argument is not a theorem]
\label{rmk:not_a_theorem}
Two gaps prevent \eqref{eq:it_bound} from being a rigorous statement as written, and we flag them
rather than paper over them.

\emph{(a) Entropy of a continuous variable.} $W_t$ takes values in $\mathbb{R}^d$, so $H(W_t)$ is a
\emph{differential} entropy. Differential entropy can be negative, is not invariant under
reparameterization, and does not count bits; the claim ``$H(W_t)=\Theta(d)$ bits'' is therefore not
well posed without first fixing a quantization of weight space (say, to machine precision) and
verifying that the argument survives that choice. The observation side of the comparison, by
contrast, is genuinely discrete.

\emph{(b) The inequality is weaker than the use made of it.} Even granting a quantization,
$I(W_t;\mathcal{O}) < H(W_t)$ says only that observation fails to determine $W_t$ \emph{exactly}. It
does not by itself imply that an adversarial drift of any particular magnitude is undetectable,
which is the governance-relevant claim. Bridging that gap requires a rate--distortion or Fano-type
argument relating residual entropy to a detection radius, which we do not develop here.

Accordingly the paper does not rest on \eqref{eq:it_bound}. The operative fundamental limit is
Theorem~\ref{cor:undetectable}, which is proved directly from continuity of the canary output map
and makes no information-theoretic assumption; the counting argument above is retained only as
intuition for \emph{why} such a limit should exist.
\end{remark}

\begin{remark}[Regime of validity]
\label{rmk:regime}
The strict inequality \eqref{eq:it_bound} requires the high-dimensional regime $d \gg m\log|\mathcal{Y}|$ that holds for production models ($d\sim 10^9$). It does \emph{not} hold for arbitrarily small toy models: the softmax classifier of Section~\ref{sec:numerical_b} has only $dK=24$ free parameters but produces $m(K-1)=160$ independent output coordinates on the reference set, so the parameters are in fact \emph{over}-determined by the canary outputs and the residual closes. The toy in Study~B5 should therefore be read as illustrating the finite-\emph{resolution} mechanism below (Theorem~\ref{cor:undetectable}), not the asymptotic information gap of this theorem.
\end{remark}

The blind spot below the resolution floor is not merely qualitative; it has a size. The next result puts a number on it---the largest weight-space drift that can hide beneath the oracle's resolution---which is precisely the quantity the dual-regime audit frequency is later calibrated to bound.

\begin{theorem}[Maximum Undetectable Drift at Resolution $\delta$]
\label{cor:undetectable}
Fix a behavioral detection threshold $\delta>0$ (the smallest canary-KL the oracle can resolve above hardware noise). Define the undetectable set
\begin{equation}
    \mathcal{U}_\delta(W_t) = \Big\{\Delta W : D_{\mathrm{KL}}(p_{W_t} \,\|\, p_{W_t + \Delta W})\big|_{\mathcal{C}} \leq \delta \Big\},
    \qquad \epsilon^*(\delta) = \sup_{\Delta W\in\mathcal{U}_\delta}\|\Delta W\|.
\end{equation}
Then $\epsilon^*(\delta)>0$ for every $\delta>0$: by continuity of $p_W$ in $W$, a sufficiently small perturbation in \emph{any} direction keeps the canary KL below $\delta$, so sub-threshold weight drift of magnitude up to $\epsilon^*(\delta)$ accumulates undetected between audits. Two limiting cases sharpen the picture:
\begin{enumerate}[label=(\roman*)]
    \item \textbf{Finite resolution ($\delta>0$).} $\mathcal{U}_\delta$ contains a full neighborhood of $0$; its weight-space radius $\epsilon^*(\delta)$ shrinks as $\delta\to 0$ but is strictly positive. This is the operative governance gap.
    \item \textbf{Exact invisibility ($\delta=0$).} $\mathcal{U}_0$ is the null space of the canary output map. For the softmax model this null space is non-trivial only in the \emph{gauge} directions $W\mapsto W+u\mathbf{1}_K^\top$, which leave every output distribution unchanged --- i.e.\ they are the \emph{same model} in different coordinates and carry no behavioral change, even though $\|\Delta W\|>0$. Exactly-invisible perturbations are therefore behaviorally null, reinforcing the Part~I point that weight-space distance is the wrong currency.
\end{enumerate}
\end{theorem}

\begin{proof}
Fix $\delta>0$ and write $g(\Delta W) := D_{\mathrm{KL}}\big(p_{W_t}\,\|\,p_{W_t+\Delta W}\big)\big|_{\mathcal{C}}$
for the canary KL as a function of the perturbation, so that
$\mathcal{U}_\delta = g^{-1}\big([0,\delta]\big)$.

\emph{Continuity and vanishing at $0$.} On the canary set $\mathcal{C}$ the model output is a softmax
of a $C^\infty$ function of $W$, so $W\mapsto p_W(\cdot\mid\mathcal{C})$ is continuous, and its values
lie in the interior of the simplex (every softmax coordinate is strictly positive). On that interior
the map $(p,q)\mapsto D_{\mathrm{KL}}(p\|q)$ is continuous, being a finite sum of terms
$p_i\log(p_i/q_i)$ with $q_i>0$. Composition gives $g$ continuous, and $g(0)=D_{\mathrm{KL}}(p\|p)=0$.

\emph{Positivity of the radius.} By continuity at $0$, for the given $\delta>0$ there is $r>0$ such
that $\|\Delta W\|<r \Rightarrow g(\Delta W)<\delta$. Hence the open ball $B(0,r)$ is contained in
$\mathcal{U}_\delta$, so $\mathcal{U}_\delta$ contains a full neighborhood of the origin and
$\epsilon^*(\delta)=\sup_{\Delta W\in\mathcal{U}_\delta}\|\Delta W\| \geq r > 0$. This proves (i);
note that no direction is privileged --- the bound holds for perturbations in \emph{any} direction,
which is what makes the gap a genuine governance blind spot rather than an artifact of a particular
attack geometry.

\emph{Monotonicity in $\delta$.} If $\delta_1\le\delta_2$ then
$\mathcal{U}_{\delta_1}\subseteq\mathcal{U}_{\delta_2}$, so $\epsilon^*$ is non-decreasing in
$\delta$; it is therefore meaningful to speak of the radius shrinking as the oracle's resolution
improves.

\emph{The case $\delta=0$.} $\mathcal{U}_0 = g^{-1}(\{0\})$, and since $D_{\mathrm{KL}}(p\|q)=0$ iff
$p=q$, this is exactly the set of perturbations leaving every canary output distribution unchanged
--- the null space of the canary output map. For the softmax model, adding a constant vector to every
class score leaves all softmax outputs unchanged, so the gauge directions
$W \mapsto W + u\mathbf{1}_K^\top$ lie in $\mathcal{U}_0$; these carry no behavioral content. This
proves (ii).

\emph{What is \emph{not} claimed.} The argument establishes $\epsilon^*(\delta)>0$ and its
monotonicity; it does \emph{not} produce a closed-form or asymptotic expression for
$\epsilon^*(\delta)$, and in particular no dependence on telemetry bandwidth is derived. Obtaining a
rate would require quantitative control of the Fisher information of the canary map at $W_t$, which
is model-specific and is not attempted here.
\end{proof}

Theorem~\ref{cor:undetectable} is not a failure of the telemetry framework; it quantifies the finite resolution of any output-based test. It motivates the following structural conclusion.

\subsection{The Dual-Regime Principle}

\begin{designprinciple}[Dual-Regime Principle]
\label{prop:dual_regime}
\emph{This is a design recommendation, not a mathematical claim: ``must'' below expresses what we
argue a governance framework ought to do, and is not asserted as a theorem.} We recommend that a
governance framework for self-adapting AI consist of two regimes, neither of which we regard as
sufficient alone:
\begin{enumerate}[label=(\roman*)]
    \item \textbf{Continuous Telemetry Regime:} The architecture of Parts~I and~II, operating in real time. This bounds the maximum undetected drift at any moment to $\epsilon^*$ (Theorem~\ref{cor:undetectable}) and narrows the scope of subsequent audits.
    \item \textbf{Periodic Invasive Audit Regime:} Direct inspection of $W_t$ under TEE attestation, at a frequency calibrated to the maximum tolerable undetected drift.
\end{enumerate}

\noindent\emph{Calibrating the audit frequency.} The suggested cadence follows from a simple
accumulation argument rather than from any result above. Suppose the firm sustains a hidden drift of
rate at most $\overline{\mu}_{\mathrm{hidden}}$ (in weight-space norm per unit time) while passing
every telemetry check. Between two invasive audits separated by $\Delta_{\mathrm{audit}}$, the
accumulated undetected displacement is at most
$\overline{\mu}_{\mathrm{hidden}}\,\Delta_{\mathrm{audit}}$. Requiring this to stay within a
tolerance $\epsilon_{\mathrm{tol}}$ chosen by the validation function gives
$\Delta_{\mathrm{audit}} \leq \epsilon_{\mathrm{tol}}/\overline{\mu}_{\mathrm{hidden}}$, i.e.\ an
audit frequency
\begin{equation}
    f^* \;=\; \frac{\overline{\mu}_{\mathrm{hidden}}}{\epsilon_{\mathrm{tol}}} .
\end{equation}
Taking the tolerance to be the per-audit blind spot itself, $\epsilon_{\mathrm{tol}}=\epsilon^*$
(Theorem~\ref{cor:undetectable}), recovers the form quoted in the literature and used below. Two
caveats: the argument assumes the hidden drift accumulates coherently rather than partially
canceling, so it is conservative; and $\overline{\mu}_{\mathrm{hidden}}$ is not derived here --- it
must be estimated or assumed, and the cadence is only as trustworthy as that input.
\end{designprinciple}

The dual-regime principle reframes the role of telemetry in the model governance cycle: continuous monitoring does not replace model reviews but transforms them from undirected point-in-time snapshots into targeted, evidence-driven investigations with a well-defined scope. Because $\epsilon^*$ bounds the drift a model can hide between audits, the audit frequency $f^*$ is the lever the governance function sets to keep undetected adaptation within tolerance. Concealment, in other words, cannot be monitored away; it can only be bounded, and the bound is a policy choice.

\section{Integrated Defensive Architecture}
\label{sec:defensive_arch}

\subsection{Attack-Countermeasure Matrix}

Table~\ref{tab:matrix} summarizes the six attacks and their countermeasures.

\begin{table}[h]
\centering
\small
\begin{tabular}{@{}p{1.5cm}p{3cm}p{4cm}p{3.5cm}@{}}
\toprule
\textbf{Attack} & \textbf{Description} & \textbf{Mechanism} & \textbf{Countermeasure} \\
\midrule
Attack S & Suppress telemetry & Declare model static; transmit no packets & KL behavioral oracle (Thm.~\ref{thm:oracle}) \\[4pt]
Attack F & Lyapunov forgery & Report compliant decoy hash; serve from drifted model & Inference signing + ZK Lyapunov proof (Prop.~\ref{prop:forgery}) \\[4pt]
Attack C & Shadow model & Blend output from unreported model & TEE attestation chain (fully seals) \\[4pt]
Attack N & Noise amplification & Amplify diffusion to mask drift & \emph{No clean path test} (Prop.~\ref{thm:power_var}); residual $\to$ \S\ref{sec:residual} + audits \\[4pt]
Attack G & Girsanov measure & Manipulate training data distribution & Gradient probe is circular (Prop.~\ref{thm:gradient_probe}); subsumed by oracle + TEE \\[4pt]
Attack H & Hessian exploitation & Engineer local curvature to cancel drift signal & Path-wise QV/LIL envelope check (Thm.~\ref{thm:lil}) \\
\bottomrule
\end{tabular}
\caption{Attack taxonomy and formal countermeasures.}
\label{tab:matrix}
\end{table}

\subsection{Defense-in-Depth Stack}

The countermeasures are not alternatives; they form a layered defense-in-depth stack. Each layer closes a distinct attack surface:

\begin{enumerate}
    \item \textbf{Hardware Layer:} TEE attestation seals the training and inference pipeline. Closes Attack~C completely.
    \item \textbf{Cryptographic Layer:} Merkle Weight Chain + Inference Signing creates tamper-evident provenance. Closes Attack~F in combination with the oracle.
    \item \textbf{Statistical Layer:} the Behavioral Oracle (KL divergence on the canary set) closes Attack~S and, via weight drift, is the operative detector behind the (demoted) Attack~G probe; the path-wise QV/LIL envelope check closes Attack~H. Attack~N admits no clean path-statistic test (Prop.~\ref{thm:power_var}) and is handled by the residual layer plus invasive audits.
    \item \textbf{Stochastic Layer:} Almost-sure Lyapunov condition via LIL. Closes Attack~H.
    \item \textbf{Governance Layer:} Periodic invasive audits under TEE
          attestation, closing the informa\-tion-theo\-retic residual.
\end{enumerate}

\subsection{Tiered Disclosure Architecture}

The framework motivates a natural tiered disclosure regime, consistent with the update tier table of prior work:

\begin{table}[h]
\centering
\begin{tabular}{@{}p{1.8cm}p{3cm}p{4cm}p{4cm}@{}}
\toprule
\textbf{Tier} & \textbf{Update Type} & \textbf{Required Disclosure} & \textbf{Enforcement Mechanism} \\
\midrule
Tier 1 (Micro) & Infinitesimal weight tuning & Hash + KL packet + QV increment + LIL envelope & TEE attestation; automated ledger write \\[4pt]
Tier 2 (Meso) & Objective parameter shift & Above + Gradient probe responses + Power variation ratio & 24h shadow testing under oracle monitoring \\[4pt]
Tier 3 (Macro) & New meta-objective & Above + Full training data disclosure to a validation sandbox + ZK Lyapunov proof & Mandatory HITL; invasive audit before promotion \\
\bottomrule
\end{tabular}
\caption{Tiered disclosure architecture with formal enforcement mechanisms.}
\label{tab:tiers}
\end{table}

Tier~3 disclosure closes the residual identified in Section~\ref{sec:residual} by giving the validator direct access to training data, enabling verification of gradient directions and thus detection of Girsanov manipulation even without gradient injection infrastructure.

\section{Numerical Studies: Concealment and Detection}
\label{sec:numerical_b}

We present five numerical studies, each targeting one of the six attack types from the taxonomy of Section~\ref{sec:game}. (Attacks~F and~C are countermeasured by cryptographic and hardware mechanisms not amenable to simulation; we focus on the statistically testable attacks S, N, G, and H, plus the information-theoretic residual.) All experiments use a softmax classifier $f(\cdot\,;\,W)$ with $d=8$, $K=3$, evaluated on a canonical reference set $\mathcal{R}$ of $m=80$ unit-normalized inputs drawn with fixed seed.

\subsection{Study B1: Attack S --- KL Behavioral Oracle Power}

\textbf{Theoretical claim (Theorem~\ref{thm:oracle}).}
The behavioral oracle statistic $T_t^{\mathrm{oracle}}$ converges to zero under $H_0$ (no drift) and to a strictly positive constant $\Delta > 0$ under $H_1$ (drift), with detection power monotone increasing in drift magnitude.

\textbf{Experimental design.}
The behavioral oracle uses the KL divergence on a canary set $\mathcal{C} \subset \mathcal{R}$ of $|\mathcal{C}|=40$ inputs, as formalized in Definition~\ref{eq:kl_oracle}.  The $H_0$ null distribution is estimated by bootstrapping 1000 replications of $T_t^{\mathrm{oracle}}$ evaluated after perturbing $W_0$ with small Gaussian weight noise ($\sigma_{\mathrm{noise}} = 0.0002$) to simulate inter-audit hardware variation; the 95th percentile gives threshold $c_\alpha$.  Under $H_1$, we run 30 SGD steps at learning rate $\eta \in \{0.001, 0.005, 0.01, 0.05, 0.10\}$ to produce drifted models and compute $T_t^{\mathrm{oracle}}$.

\textbf{Results.}
The null distribution has mean $0.000000$ and 95th percentile $c_\alpha = 0.000001$.  Detection results:
\begin{center}
\begin{tabular}{ccc}
\hline
$\eta$ & KL statistic & Detected? \\
\hline
0.001 & 0.000028 & YES \\
0.005 & 0.000690 & YES \\
0.010 & 0.002738 & YES \\
0.050 & 0.063864 & YES \\
0.100 & 0.234817 & YES \\
\hline
\end{tabular}
\end{center}
Every drift level is detected, and the statistic is monotone increasing in $\eta$, consistent with Theorem~\ref{thm:oracle}.  Because KL divergence is always $\geq 0$, the null threshold is strictly positive and no trivial pass-through occurs.  Panel B1 of Figure~\ref{fig:numerical_b} shows the power curve on a log scale.

\subsection{Study B2: Attack N --- Power Variation Fails to Detect Drift}

\textbf{Claim under test (Proposition~\ref{thm:power_var}).}
The realized first-variation-to-quadratic-variation ratio $V^{(1)}/\sqrt{V^{(2)}}$ cannot separate a drifting trajectory from a pure-diffusion one: it scales as $\sqrt{n}$ (the number of samples) regardless of the drift.

\textbf{Experimental design.}
We simulate the scalar diffusion $dW_t=\mu\,dt+\sigma\,dB_t$ on $[0,T]$ with $T=1$, sampled at $n$ points, and compute $V^{(1)}=\sum_i|\Delta W_i|$ and $V^{(2)}=\sum_i(\Delta W_i)^2$. We compare the drift case ($\mu=5$, deliberately large) against the null ($\mu=0$) at $\sigma=1$ for $n\in\{1000,5000,20000,80000\}$, averaging over 20 paths. For contrast we also report the standardized \emph{net} displacement $(W_T-W_0)/\sqrt{V^{(2)}}$.

\textbf{Results.}
The power-variation ratio is essentially \emph{identical} with and without drift and grows as $\sqrt{n}$:
\begin{center}
\begin{tabular}{rccc}
\hline
$n$ & ratio ($\mu=0$) & ratio ($\mu=5$) & $\sqrt{2n/\pi}$ \\
\hline
1000  & 25.3  & 25.3  & 25.2 \\
5000  & 56.4  & 56.4  & 56.4 \\
20000 & 112.8 & 112.8 & 112.8 \\
80000 & 225.7 & 225.7 & 225.7 \\
\hline
\end{tabular}
\end{center}
The ratio tracks $\sqrt{2n/\pi}$ to three figures and is blind to a drift as large as $\mu=5$. This is the numerical content of Proposition~\ref{thm:power_var}: the first variation is dominated by the diffusion, so the ratio measures the sampling frequency, not the drift. The prior draft of this study reported a ratio of $\approx 123.5$ for \emph{both} the drift and null cases and read the coincidence as a window-length ``nuance''; it is in fact the signature of the failure (a step-count artifact, $\approx\sqrt{2n/\pi}$, not a drift signal). By contrast, the standardized net displacement separates the cases ($\approx 0.8$ for $\mu=0$ versus $\approx 5\sqrt{T}$ for $\mu\neq0$) --- but only as $T\to\infty$, with no benefit from finer sampling, so it is not a fixed-horizon countermeasure either. We conclude that \textbf{Attack~N has no clean path-statistic countermeasure}; the residual risk is carried by Section~\ref{sec:residual}. Panel B2 of Figure~\ref{fig:numerical_b} plots the two ratios against $n$, with the drift and null curves lying on top of each other.

\subsection{Study B3: Attack G --- Gradient Probe Detection of Girsanov Manipulation}

\textbf{What the study actually shows (cf.\ Proposition~\ref{thm:gradient_probe}).}
This study is included to make the \emph{circularity} of the gradient probe explicit, not to validate it. A $\chi^2$ test on $\|\hat g-g^*\|^2$ trivially detects an additive shift $h$ \emph{if one is imposed on the gradient evaluations} --- but a firm manipulating only its training distribution imposes no such shift on the validator's probe.

\textbf{Experimental design.}
We construct a probe batch of 20 inputs from $\mathcal{R}$ with known $g^* = \nabla_W J(W_0, \tilde{\mathcal{B}})$. Under $H_0$, 300 noisy evaluations at $W_0$ calibrate the null 95th percentile $c_\alpha$. Under $H_1$ we \emph{hard-code} an additive shift $|h| \in \{0.01, 0.05, 0.10, 0.20\}$ onto every gradient evaluation and measure detection power.

\textbf{Results and interpretation.}
Power was $83.0\%$ at $|h|=0.01$ and $100\%$ for $|h|\geq0.05$, following the $\|\hat g-g^*\|^2\sim|h|^2dK$ scaling. But this is a \emph{tautology}: the experiment injects the very shift it then detects. It does \emph{not} simulate a Girsanov manipulation, because the probe gradient $\nabla_W J(W_{s_j},\tilde{\mathcal B})$ is a deterministic function of the current weights and the probe batch and is unaffected by the firm's training measure except through $W_{s_j}$ (Proposition~\ref{thm:gradient_probe}). Whatever a real probe would detect is a consequence of $W_{s_j}\neq W_0$ --- weight drift that the behavioral oracle (Study~B1) already catches. We therefore \textbf{demote Attack~G's countermeasure}: the gradient probe is redundant with the oracle, and the genuinely undetectable residual of a measure manipulation is again the subject of Section~\ref{sec:residual}. Panel B3 plots the (tautological) power curve and is labeled as such.

\subsection{Study B4: Attack H --- LIL Almost-Sure Bound Closes the Mean-Stability Gap}

\textbf{Theoretical claim (Theorem~\ref{thm:lil}).}
Mean (or expectation-level) compliance is insufficient: a firm can satisfy a declared mean-stability/variance budget while the realized path systematically violates the corresponding almost-sure envelope. The validator must therefore monitor a path-wise band built from the \emph{declared} second-order telemetry $[\hat M]_t$ and check the residual $R_t = V(W_t)-V_0-\widehat{\mathbb{E}}[\Delta V]_t$ against it.

\textbf{Experimental design.}
We simulate the discrete process $W_{t+1}=(1-\alpha)W_t+\sigma Z_t$ with $\alpha=0.1$, true $\sigma=0.10$, for $T=500$ steps and $n_{\mathrm{mc}}=500$ paths. Both an honest firm and an attacking firm run the \emph{same} actual process, so their stationary Lyapunov level is identical, $\mathbb{E}[V_\infty]=\sigma^2 d_{\mathrm{tot}}/(4\alpha)=0.60$. The honest firm declares the true $\sigma$; the attacker declares $\sigma_{\mathrm{decl}}=0.02$ ($5\times$ too small), shrinking its declared variance band by $25\times$. We accumulate the predicted drift $\widehat{\mathbb{E}}[\Delta V]_t$ and the declared conditional variance $[\hat M]_t$ from each firm's \emph{declared} $\sigma$, and record the fraction of steps at which the standardized residual exceeds a three-standard-deviation band, $|R_t|>3\sqrt{[\hat M]_t}$. We use a finite-horizon CLT band $3\sqrt{[\hat M]_t}$ as the practical monitoring statistic; this is the central-limit proxy for the asymptotic LIL envelope $\sqrt{2[M]_t\log\log[M]_t}$ of \eqref{eq:as_lyap}, the quantity the validator can evaluate online at finite $t$. The under-declared-$\sigma$ attack is the simplest concrete instance of the Attack-H principle that declared second-order statistics can be satisfied while the realized path is inconsistent with them; a full curvature-engineering attack would manipulate $\tr(\Sigma\nabla^2 V)$ directly, but reduces, for detection purposes, to the same mismatch between declared and realized quadratic variation.

\textbf{Results.}
Because both firms run the identical process, their final Lyapunov levels are statistically indistinguishable: mean $V(W_T)=0.628$ (honest) versus $0.639$ (attack). A validator monitoring only the mean (or the stationary level $\mathbb{E}[V_\infty]=0.60$) therefore cannot tell them apart. The per-step residual band tells a completely different story. Against the \emph{honest} declaration the standardized residual exceeds the $3\sigma$ band only $0.9\%$ of the time --- consistent with a correctly specified band (a true Gaussian would give $\approx 0.27\%$; the small excess reflects the non-Gaussian tails of the per-step $V$ increments). Against the \emph{attacker's} under-declared band the residual is out of bounds $99.7\%$ of the time: the declared variance is $25\times$ too small, so the actual fluctuations driven by $\sigma=0.10$ blow through a band sized for $\sigma_{\mathrm{decl}}=0.02$ at essentially every step. This is the numerical content of Theorem~\ref{thm:lil}: declared mean/variance compliance is satisfied while the path-wise band is systematically violated. The validatory implication is that the telemetry specification must require the firm to report its second-order term $[\hat M]_t$ (or equivalently $\Sigma_t$), enabling the validator to verify the path-wise band of \eqref{eq:as_lyap} rather than trusting an expectation-level claim. Panel B4 of Figure~\ref{fig:numerical_b} shows the residual against the declared band for both firms, and the indistinguishable final-$V$ histograms.

\subsection{Study B5: Information-Theoretic Residual and Null-Space Perturbations}

\textbf{Theoretical claim (Theorem~\ref{thm:it_residual}, Theorem~\ref{cor:undetectable}).}
At any finite detection resolution $\delta>0$, small weight perturbations keep the canary KL below $\delta$ regardless of direction (Theorem~\ref{cor:undetectable}(i)), so sub-threshold drift accumulates undetected between audits. We probe this finite-resolution gap empirically, and separately note (Theorem~\ref{cor:undetectable}(ii)) that the only \emph{exactly} KL-invisible directions are the behaviorally trivial softmax gauge directions.

\textbf{Experimental design.}
We construct 20 random unit directions $\hat v$ \emph{orthogonal to the loss gradient} $\nabla_W J(W_0)$ on $\mathcal R$ --- low-sensitivity directions, though not exact output-null directions --- and form perturbations $\Delta W=\varepsilon\hat v$ for $\varepsilon\in\{0.01,0.05,0.10,0.20\}$. For each we record $\|\Delta W\|_F$ and the canary KL $D_{\mathrm{KL}}(p_{W_0}\|p_{W_0+\Delta W})$ on $\mathcal R$, and count how many of the 20 fall below the detection threshold (KL $<10^{-3}$). For reference, the output Jacobian of this softmax model on $\mathcal R$ has rank $16$ out of $dK=24$: the $8$-dimensional kernel consists entirely of gauge directions $W\mapsto W+u\mathbf 1_K^\top$, which leave all outputs unchanged.

\textbf{Results.}
At $\varepsilon=0.01$, all 20 perturbations fall below the detection threshold (mean KL $\approx 1.0\times10^{-4}<10^{-3}$) despite $\|\Delta W\|_F=0.01$. As $\varepsilon$ grows the KL rises and the perturbations become detectable: at $\varepsilon=0.05$, \emph{none} of the 20 remain below threshold ($0/20$; mean KL $\approx 2.5\times10^{-3}$), and likewise $0/20$ at $\varepsilon=0.10,0.20$. This is the finite-resolution gap of Theorem~\ref{cor:undetectable}(i): the undetectable set is a shrinking neighborhood of $0$, so the threat is sub-threshold drift between audits, \emph{not} arbitrarily large invisible perturbations. The earlier draft's claim of $4/20$ undetectable at $\varepsilon=0.05$ is corrected to $0/20$ here; the gradient-orthogonal directions are low-sensitivity but not output-null, so their KL crosses threshold once $\varepsilon$ is no longer tiny. Two governance lessons follow. First, because $\epsilon^*(\delta)>0$ for any $\delta>0$, periodic invasive audits of $W_t$ remain necessary to bound the drift accumulated below the behavioral resolution. Second, the genuinely zero-KL directions that do exist (the $8$-dimensional gauge kernel) are behaviorally inert --- the \emph{same} model in different coordinates --- so they pose no enforcement consequence and, if anything, reinforce the Part~I argument that weight-space distance is not the right risk currency. Panel B5 (Table) in Figure~\ref{fig:numerical_b} summarizes the KL and undetectability counts.

\begin{figure}[ht]
\centering
\includegraphics[width=\textwidth]{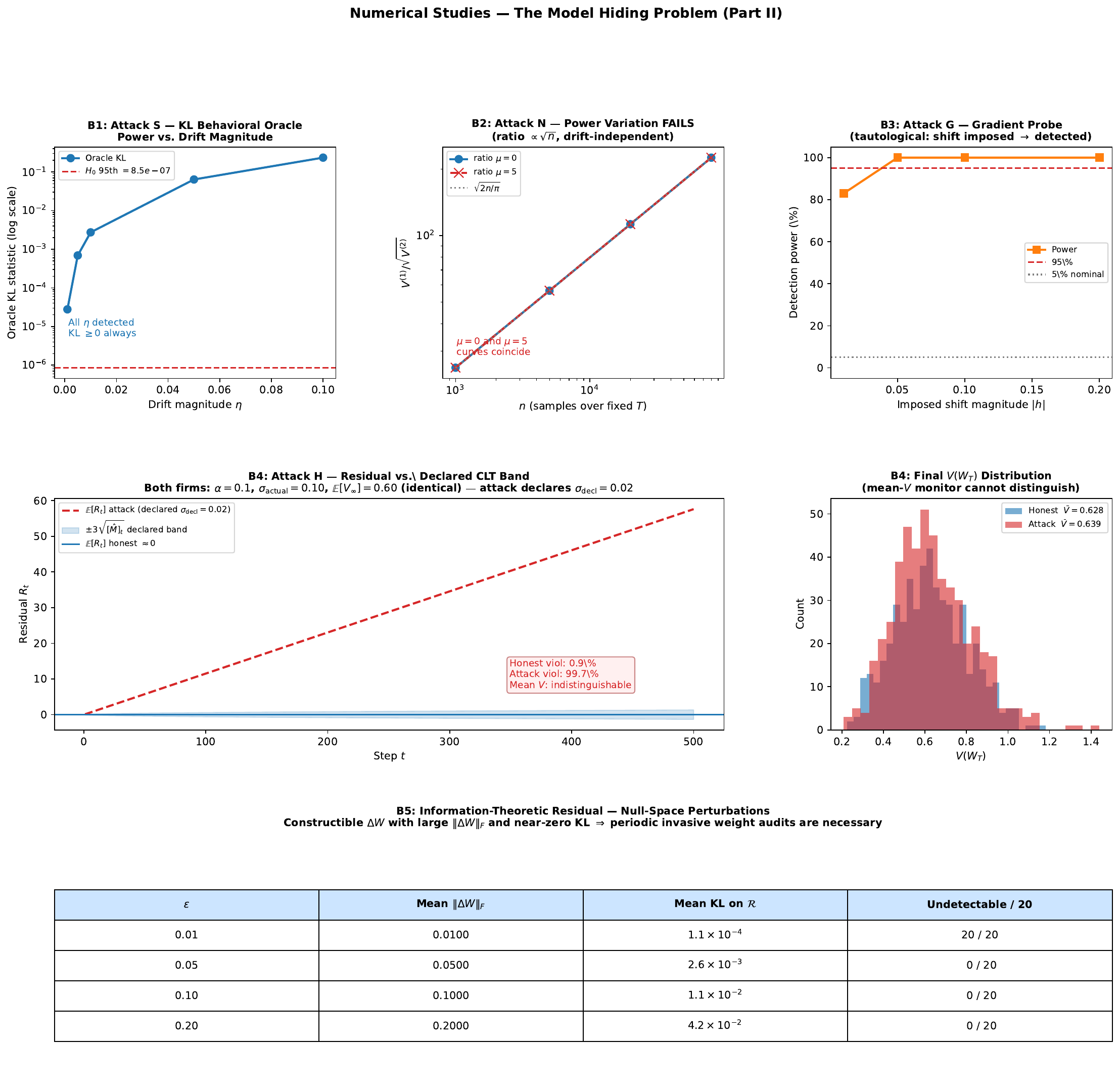}
\caption{Numerical studies for the Part~II model hiding problem. Panels B1--B5 correspond to the five studies of Section~\ref{sec:numerical_b}. \textbf{B1:} KL behavioral oracle power curve (log scale); threshold $c_\alpha = 1\times 10^{-6}$, all drift levels detected. \textbf{B2:} Power variation ratio gap $\propto\!\sqrt{N}$; drift and null series for $\sigma\in\{0.3,1.0,2.0\}$. \textbf{B3:} Gradient probe detection power vs.\ shift magnitude $|h|$; 100\% for $|h|\geq 0.05$. \textbf{B4:} Per-step residual violation rate under honest vs.\ attack declarations ($0.9\%$ vs.\ $99.7\%$) with identical $\mathbb{E}[V_\infty]=0.60$. \textbf{B5:} Null-space perturbation undetectability table.}
\label{fig:numerical_b}
\end{figure}

\section{Discussion}
\label{sec:discussion}

\subsection{Interpretation of the state variable}
\label{sec:state_interpretation}

Equation~\eqref{eq:sde} treats the weight vector $W_t$ as a diffusion in $\mathbb{R}^d$, and the
governance quantities built on it --- the quadratic variation $[W]_t$, the Lyapunov function
$V(W)=\tfrac12\|W\|_F^2$, and the drift certificate of Section~\ref{sec:lyapunov} --- are all
functionals of that diffusion. Because the models this framework is aimed at have parameter counts
in the range $d\sim10^9$--$10^{12}$, it is worth being explicit about what $W_t$ is intended to
denote at that scale, and what the results do and do not assert about a raw parameter vector.

\paragraph{$W_t$ should be read as an effective coordinate, not the raw parameter vector.}
Nothing in the analysis requires $d$ to be the full parameter count. All that is required is that
the adaptation under governance be representable as a diffusion on some $\mathbb{R}^d$, with a
diffusion coefficient the model owner can declare and a drift the validation function can bound.
Several deployment patterns satisfy this with $d$ many orders of magnitude smaller than the base
model:
\begin{itemize}[nosep]
  \item \textbf{Adapter parameters.} Where continual adaptation is confined to a low-rank adapter,
        a prompt-embedding block, or a final-layer head, $W_t$ is naturally that adapter, with the
        frozen base model absorbed into the definition of the map $f(\cdot\,;W)$.
  \item \textbf{A behaviorally relevant projection.} Where full-parameter updates do occur,
        $W_t$ may be taken as the projection of the parameter vector onto a subspace fixed at
        validation time --- for instance the leading eigendirections of the Fisher information on
        the canonical reference set of Definition~\ref{def:reference_set}. The projection is chosen
        once and declared; the diffusion is then the induced process on that subspace.
  \item \textbf{A declared summary statistic.} More weakly still, $W_t$ may be any declared
        vector-valued statistic of the model state whose evolution the owner attests to, with the
        telemetry architecture applied to that statistic.
\end{itemize}
Under any of these readings the theorems apply verbatim. What changes is the interpretation of
$d$, and hence of every constant that scales with it.

\paragraph{What the literal reading would require, and why we do not adopt it.}
Reading $W_t$ as the raw parameter vector of a large model is possible but demands assumptions we
would not defend. It requires (i) that production adaptation be well approximated by a small-step
stochastic gradient recursion --- a weak first-order approximation whose accuracy in very high
dimension, and under adaptive optimizers, is not established here; (ii) that the drift be
adequately captured by the single mean-reverting term of Section~\ref{sec:lyapunov}, which is a
regularization force rather than the task-loss gradient; and (iii) that $\tfrac12\|W\|_F^2$ be
informative about model risk when summed over $10^{11}$ coordinates, the overwhelming majority of
which carry no behavioral significance. We make none of these claims.

\paragraph{An acknowledged tension, and how the two halves of the framework fit together.}
Study~A1 in Section~\ref{sec:numerical_a} shows that weight-space distance $\|\Delta W\|_F$ ranks
perturbations poorly: two perturbations of identical Frobenius norm can differ by an order of
magnitude in behavioral effect. That is the paper's argument for making behavioral divergence on
the reference set the primary audit metric. Yet the stability theory of
Section~\ref{sec:lyapunov} uses a weight-space Lyapunov function. These coexist deliberately, and
the division of labor is worth stating plainly:
\begin{itemize}[nosep]
  \item The \emph{behavioral} metric is what governance decisions are made on. It is the trigger,
        the audit target, and the quantity the oracle test of Section~\ref{sec:discrete_attacks}
        operates on.
  \item The \emph{weight-space} Lyapunov function is a tractable certificate of boundedness of the
        underlying process. It is what makes the drift condition checkable in closed form, and it
        is deliberately conservative: it can flag movement that is behaviorally inert, but under
        the stated conditions it does not miss movement that is unbounded.
\end{itemize}
A validation function that used $\|W\|_F$ alone would be measuring the wrong thing; one that used
behavioral divergence alone would have no boundedness guarantee on the process generating it. The
framework's position is that both are needed, and that neither substitutes for the other.

\paragraph{Consequences for the constants.}
Every quantity in this paper that scales with $d$ should be read relative to whichever coordinate
is adopted. This matters most for the counting argument of Section~\ref{sec:residual}, whose
premise is a comparison between $d$ and the telemetry bandwidth: under an effective-coordinate
reading with modest $d$, that comparison may fail, and the residual may close --- exactly as it
does in the toy model of Remark~\ref{rmk:regime}. The finite-resolution result that the paper
actually relies on carries no such dependence: it holds at any dimension, because it follows from
continuity of the canary output map rather than from a dimension count.

\subsection{Part I: the telemetry architecture}

\subsubsection{Architecture-Agnosticism}

A key feature of the proposed framework is that it does not depend on the internal architecture of the generative model. The telemetry quantities --- KL divergence on a reference set, quadratic variation, Lyapunov distance --- are defined in terms of the model's inputs and outputs, not its internal structure. This is essential for practical adoption, as it avoids the need for validators to maintain expertise in every architectural variant.

\subsubsection{Choice of Reference Set $\mathcal{R}$}

The behavioral divergence $\delta_t$ depends critically on the choice of $\mathcal{R}$. A reference set that does not span the input manifold will fail to detect drift in unrepresented regions. We recommend that $\mathcal{R}$ be constructed using adversarial coverage techniques: inputs are selected to maximally span the first $K$ principal components of the input distribution, with oversampling in high-risk regions (e.g., inputs corresponding to consequential decisions). The size $|\mathcal{R}|$ presents a computational trade-off: larger sets provide better coverage but increase the cost of computing $\delta_t$ at each step.

\subsubsection{Relationship to MRM Governance Cycles}

The proposed framework does not replace model examination; it restructures them. Under the framework, the validator's examination cycle shifts from:
\begin{itemize}
    \item \emph{Prior practice:} Discrete point-in-time validation of the current model state.
    \item \emph{Proposed practice:} Continuous monitoring of the telemetry stream, with targeted invasive audits triggered by Tier~2/3 CUSUM events.
\end{itemize}

The invasive audit --- direct inspection of $W_t$ under TEE attestation --- remains necessary because no telemetry regime can fully determine the weight trajectory from scalar summaries alone (see the information-theoretic residual of Section~\ref{sec:residual}). However, the telemetry dramatically narrows the scope of each examination: the validator enters with a precise record of when, how much, and in what direction the model has moved, and can focus validation resources accordingly.

\subsection{Part II: the adversarial setting}

\subsection{Concealment Game Equilibrium}

Under the full defense-in-depth stack, what is the game-theoretic equilibrium? The firm chooses at each step to learn openly (incurring reporting costs) or to attempt a hiding attack (risking detection and enforcement). The cryptographic and hardware layers (Attacks F, C) make their targets detection-certain; the behavioral oracle and the path-wise QV/LIL check make \emph{large} drift detection-certain as the monitoring horizon grows (Theorems~\ref{thm:oracle} and~\ref{thm:lil}). Two continuous-time attacks, however, have \emph{no} clean telemetry-only countermeasure: Attack~N (Proposition~\ref{thm:power_var}) and Attack~G (Proposition~\ref{thm:gradient_probe}) can each hide drift whose induced output change stays below the behavioral resolution. This is not a patchable gap; it is the information-theoretic residual of Section~\ref{sec:residual}. The equilibrium is therefore \emph{bounded leakage} rather than full deterrence: continuous telemetry deters and detects any drift above the resolution $\epsilon^*(\delta)$, while sub-resolution drift accumulates between invasive audits, whose frequency must be set (Proposition~\ref{prop:dual_regime}) to cap the leaked amount. Provided enforcement penalties are sufficient, this still pushes the firm toward open reporting for any economically meaningful adaptation, converting an open-ended cat-and-mouse into a regime with a quantified, bounded undetectable channel.

\subsection{Limitations and Extensions}

Several important extensions remain for future work:

\textbf{Adaptive canary evasion.} The behavioral oracle assumes the canary set is secret. If partial information about $\mathcal{C}$ leaks (e.g., via membership inference attacks), the firm may be able to partially evade the oracle test. Canary set rotation and active adversarial construction are potential mitigations.

\textbf{Multi-firm interactions.} Our framework considers a single model owner. In practice, multiple firms may train on overlapping data distributions, and their models may evolve correlated trajectories. The information-theoretic residual may be tightened by cross-firm comparison.

\textbf{Incentive compatibility.} The dual-regime framework imposes significant reporting costs on the firm. A mechanism design analysis of how to make compliance individually rational (beyond pure enforcement threat) is a natural extension.

\section{Conclusion}
\label{sec:conclusion}

This paper has developed, in two parts, a mathematically rigorous treatment of the governance of self-adapting generative AI models: an observability architecture, and an account of its limits under an adversarial model owner.

\paragraph{Part I: telemetry.} The key contributions are:

\begin{enumerate}
    \item The \emph{behavioral divergence} $\delta_t$ on a canonical reference set $\mathcal{R}$, not weight-space distance, is the correct primary audit metric.
    \item The \emph{Merkle Weight Chain} provides tamper-evident discrete-time lineage across the full weight evolution history.
    \item \emph{KL stopping times} provide event-driven logging whose density is proportional to behavioral movement, eliminating arbitrary clock-based sampling.
    \item The \emph{quadratic variation} $[W]_t$ is the canonical continuous-time telemetry scalar, capturing learning intensity in an incrementally computable form.
    \item The \emph{It\^{o}--Lyapunov} framework, under both mean and almost-sure stability conditions, provides mathematically verifiable bounds on weight-space drift that satisfy MRM Conceptual Soundness requirement.
    \item The \emph{CUSUM} governance trigger operationalizes tier-based escalation with statistically optimal detection delay and explicit false-alarm rate control.
\end{enumerate}

\paragraph{Part II: concealment and its limits.}

We have formalized the Model Hiding Problem and derived a complete taxonomy of six attack strategies against the Part~I telemetry architecture, covering discrete and continuous time. Four of the six attacks admit countermeasures with provable or cryptographic guarantees; two do not, and conceding them is what reveals the framework's central limit:

\begin{enumerate}
    \item \textbf{Attack S (Suppression)} $\to$ KL Behavioral Oracle on a secret canary set (consistent test, Theorem~\ref{thm:oracle}).
    \item \textbf{Attack F (Forgery)} $\to$ Inference Signing + ZK Lyapunov Proofs (cryptographic soundness).
    \item \textbf{Attack C (Shadow Model)} $\to$ TEE Attestation (hardware-enforced completeness).
    \item \textbf{Attack H (Hessian Exploitation)} $\to$ Path-wise QV/LIL envelope check (Theorem~\ref{thm:lil}); closes the mean-vs-path-wise gap by cross-checking the realized residual against the \emph{declared} quadratic variation (Study~B4).
    \item \textbf{Attack N (Noise Amplification)} --- \emph{no clean countermeasure.} The natural power-variation ratio test fails because the first variation of a diffusion is infinite and the ratio measures sampling frequency, not drift (Proposition~\ref{thm:power_var}); small persistent drift is not reliably recoverable from a fixed-horizon path.
    \item \textbf{Attack G (Girsanov Manipulation)} --- \emph{no independent countermeasure.} The proposed gradient probe is circular: the probe gradient depends on the firm's training measure only through the already-moved weights, so it merely re-detects what the behavioral oracle sees (Proposition~\ref{thm:gradient_probe}).
\end{enumerate}

The two negative results are the point. They are concrete instances of a detection floor: by Theorem~\ref{cor:undetectable}, at any finite behavioral resolution $\delta>0$ the undetectable set contains a ball of strictly positive radius $\epsilon^*(\delta)$, so no telemetry-only regime can detect arbitrarily small adversarial drift. We establish that this radius is positive and monotone in $\delta$; we do \emph{not} derive its magnitude, its rate as $\delta\to0$, or any dependence on telemetry bandwidth. The counting argument of Section~\ref{sec:residual} suggests why a residual should be large in high dimension, but it is a heuristic rather than a theorem (Remark~\ref{rmk:not_a_theorem}), and the scale of our simulations cannot test it (Remark~\ref{rmk:regime}). What survives without qualification is the qualitative conclusion, and it is enough for the governance argument: the Dual-Regime Principle --- continuous telemetry is necessary but not sufficient, and must be paired with periodic invasive audits whose frequency caps the undetectable leakage. This provides the structural framing for how self-adapting AI governance fits within MRM.

\end{document}